\RequirePackage[l2tabu,orthodox]{nag}
\documentclass
[11pt,letterpaper]
{article} 

\usepackage[notes=true,later=false,camera=false]{dtrt}
\usepackage[utf8]{inputenc}
\usepackage{etex}
\usepackage{ stmaryrd }
\usepackage{xspace,enumerate}
\usepackage[T1]{fontenc}
\usepackage[full]{textcomp}
\usepackage[american]{babel}
\usepackage{underscore}
\usepackage{mathtools}

\usepackage{amsthm}
\usepackage{thm-restate}
\usepackage{empheq}

\usepackage{hyperref}
\usepackage[pdftex]{graphicx}

\usepackage[capitalise,nameinlink]{cleveref}
\crefname{lemma}{Lemma}{Lemmas}
\crefname{fact}{Fact}{Facts}
\newcommand{\colorconstraints}{\text{Color Constraints}}
\crefname{colorconstraints}{(color constraints)}{Color Constraints}
\crefformat{colorconstraints}{#2\colorconstraints#3}
\crefname{indsetconstraints}{(indset constraints)}{IndSet Constraints}
\crefformat{indsetconstraints}{#2$\mathsf{IndSet\ Axioms}$#3}
\crefname{theorem}{Theorem}{Theorems}
\crefname{mtheorem}{Theorem}{Theorems}
\crefname{corollary}{Corollary}{Corollaries}
\crefname{claim}{Claim}{Claims}
\crefname{example}{Example}{Examples}
\crefname{algorithm}{Algorithm}{Algorithms}
\crefname{problem}{Problem}{Problems}
\crefname{definition}{Definition}{Definitions}
\usepackage{paralist}
\usepackage{turnstile}
\usepackage{mdframed}
\usepackage{tikz}
\usepackage{caption}
\DeclareCaptionType{Algorithm}
\usepackage{newfloat}
\newtheorem{theorem}{Theorem}[section]
\newtheorem*{theorem*}{Theorem}

\newtheorem*{proposition*}{Proposition}
\newtheorem{prop}[theorem]{Proposition}
\newtheorem{lemma}[theorem]{Lemma}
\newtheorem*{lemma*}{Lemma}

\newtheorem{coro}[theorem]{Corollary}
\newtheorem*{conjecture*}{Conjecture}

\newtheorem*{fact*}{Fact}

\newtheorem*{hypothesis*}{Hypothesis}

\theoremstyle{definition}

\newtheorem*{definition*}{Definition}

\newtheorem{algorithm}[theorem]{Algorithm}

\theoremstyle{remark}

\newtheorem*{claim*}{Claim}
\newtheorem{remark}[theorem]{Remark}
\newtheorem*{remark*}{Remark}

\newtheorem*{observation*}{Observation}

\usepackage[
letterpaper,
top=1.2in,
bottom=1.2in,
left=1in,
right=1in]{geometry}
\usepackage{newpxtext} 
\usepackage{textcomp} 
\usepackage[varg,bigdelims]{newpxmath}
\usepackage[scr=rsfso]{mathalfa}
\usepackage{bm} 
\let\mathbb\varmathbb
\usepackage{microtype}

\allowdisplaybreaks

\newcommand{\R}{{\mathbb R}}

\newcommand{\norm}[1]{\lVert #1 \rVert}

\let\epsilon=\varepsilon
\newcommand{\eps}{\varepsilon}

\newcommand{\E}{{\mathbb E}}

\newcommand{\poly}{\mathrm{poly}}

\newcommand{\vol}{\mathrm{vol}}

\newcommand{\size}{\mathrm{size}}

\begin{document}

\title{Beyond Moments: Robustly Learning Affine Transformations with Asymptotically Optimal Error}
\author{He Jia\thanks{Supported in part by NSF awards CCF-2007443 and CCF-2106444.}\\ hjia36@gatech.edu\\ Georgia Tech \and Pravesh K . Kothari\thanks{\texttt{praveshk@cs.cmu.edu}. Supported by NSF CAREER Award \#2047933, NSF \#2211971, an Alfred P. Sloan Fellowship, and a Google Research Scholar Award.}\\ praveshk@cmu.edu \\ CMU \and Santosh S. Vempala\footnotemark[1]\\ vempala@gatech.edu\\ Georgia Tech}


\maketitle

\begin{abstract}
We present a polynomial-time algorithm for robustly learning an unknown affine transformation of the standard hypercube from samples, an important and well-studied setting for independent component analysis (ICA). Specifically, given an $\epsilon$-corrupted sample from a distribution $D$ obtained by applying an unknown affine transformation $x \rightarrow Ax+s$ to the uniform distribution on a $d$-dimensional hypercube $[-1,1]^d$, our algorithm constructs $\hat{A}, \hat{s}$ such that the total variation distance of the distribution $\hat{D}$ from $D$ is $O(\epsilon)$ using poly$(d)$ time and samples. Total variation distance is the information-theoretically strongest possible notion of distance in our setting and our recovery guarantees in this distance are optimal up to the absolute constant factor multiplying $\epsilon$. In particular, if the columns of $A$ are normalized to be unit length, our total variation distance guarantee implies a bound on the {\em sum} of the $\ell_2$ distances between the column vectors of $A$ and $A'$, $\sum_{i =1}^d \norm{a_i-\hat{a}_i}_2 = O(\epsilon)$. In contrast, the strongest known prior results only yield a $\epsilon^{O(1)}$ (relative) bound on the distance between \emph{individual} $a_i$'s and their estimates and translate into an $O(d\epsilon)$ bound on the total variation distance.

Prior algorithms for this problem rely on implementing standard approaches~\cite{Cardoso1998multidimensional} for ICA based on the classical \emph{method of moments}~\cite{FJK:96,NguyenR09} combined with robust moment estimators. We observe that an approach based on $o(\log d)$-degree moments provably fails to obtain non-trivial total variation distance guarantees for robustly learning an affine transformation unless $\epsilon < 1/d^{O(1)}$. Our key innovation is a new approach to ICA (even to outlier-free ICA) that circumvents the difficulties in the classical method of moments and instead relies on a new geometric \emph{certificate} of correctness of an affine transformation. Our algorithm is based on a new  method that iteratively improves an estimate of the unknown affine transformation whenever the requirements of the certificate are not met. 
\end{abstract}

\thispagestyle{empty}
\newpage
\pagenumbering{arabic}
\newpage
\section{Introduction}
We consider the problem of learning affine transformations from samples. Specifically,  we are given i.i.d. points $x \in \R^d$ obtained after applying an unknown affine transformation to a uniform sample from the hypercube $[-1,1]^d$, i.e., $x=As+a$ where $a,A$ are unknown and each coordinate of $s$ is uniformly sampled in $[-1,1]$. The study of efficient algorithms for estimating the unknown affine transformation up to desired error is a major topic in signal processing~\cite{Cardoso1998multidimensional,Comon91,ComonJutten}, with many interesting algorithms and heuristics. It is often called \emph{standard ICA} (a special case of the well-studied \emph{Independent Component Analysis}) or {\em blind deconvolution} or the ``cocktail party" problem. 

Algorithms for recovering the unknown affine transformation $A$ are generally based on higher directional moments of the distribution. The empirical mean and covariance of the transformed samples can be used to find an affine transformation that matches the mean and second moments of the original cube. The correct rotation can be identified by first making the distrbution isotropic (zero mean, identity covariance) and then examining the fourth moment of the empirical distribution. The directions that maximize the fourth moment correspond to the facet normals of the correct rotation, see e.g., ~\cite{FJK:96,NguyenR09}. The ``method of moments" has been extended, using higher moments, to various generalizations of standard ICA, including more general product distributions and underdetermined ICA~\cite{GVX14}.  

While the model has been quite influential and is widely studied, it is reasonable to expect that data will contain errors and will deviate, at least slightly, from the precise model. Recovering the underlying model parameters despite corruption, even arbitrary adversarial noise, is the mainstay of robust statistics, a field that has enjoyed a renaissance over the past decade (and is now called Algorithmic Robust Statistics). Beginning with the robust estimation of the mean of high-dimensional distributions~\cite{DKKLMS16,LaiRV16}, there has been tremendous progress on a variety of well-known and central problems in statistical learning theory, including linear regression~\cite{KlivansKM18,bakshi2021robust}\, covariance estimation~\cite{CDGW19} and Gaussian mixture models~\cite{bakshi2022robustly}. In all these cases, nearly optimal guarantees are known, asymptotically matching statistical lower bounds for the error of the estimated parameters. 

Despite much progress on ICA and on robust estimation, the  {\em robust} version of the problem has thus far evaded solution. The precise problem is as follows: we are given samples from an unknown affine transformation of a cube, after an $\eps$ fraction of the sample has been arbitrarily (adversarially) corrupted; estimate the affine transformation. Information-theoretically, it is possible to estimate an affine transformation so that the resulting distribution is within TV distance $O(\eps)$ of the unknown transformation. But can we find this algorithmically? 

Prior works~\cite{LaiRV16,KS17} obtained some guarantees for ICA in the presence of adversarial outliers by applying robust estimators for moments of data into the classical ICA algorithms based on the method of moments. The resulting guarantees allow recovering a linear transformation $\hat{A}$ so that (up to a permutation) each column of $\hat{A}$ is close to the corresponding column of $A$ up to $(1 \pm \epsilon^{O(1)})$ relative error in $\ell_2$ norm. However, as we discuss next, this guarantee is \emph{extremely} weak and implies no upper bound on the total variation distance. Indeed a total variation guarantee requires (and our methods here will obtain!) a bound of $O(\epsilon)$ on the \emph{sum} of the $\ell_2$ errors over all columns! In particular, robust ICA algorithms from prior works yield a bound on the relevant parameter distance, namely, the total $\ell_2$ error, which is off by a factor $d$ --- the underlying dimension. 

As we next discuss, this abject failure of known methods in obtaining strong recovery guarantees for ICA is in fact an inherent issue in any algorithm that relies on the method of moments and one of our main conceptual contribution is a truly new, non-method-of-moments algorithm for learning affine transformations, even in the non-robust setting. 

\paragraph{Inadequacy of the Method of Moments:} It has been shown that some of the algorithms for ICA are robust to structured noise such as Gaussian noise, i.e., instead of observing $x = As +a$, we see $x = As+a+z$ where $z\sim N(0, \sigma^2 I)$ is Gaussian~\cite{AroraGMS12,BelkinRV12}. However, adversarial noise breaks these classical methods, which are generally based on a constant number of moments. A natural idea is to replace moments with their robust counterparts, given that robust moment estimation is one of the successes of algorithmic robust statistics. However, as we illustrate next, these methods fall short for ICA.

Consider a unit cube whose center is shifted to an unknown point $\mu$. Now the problem simply consists of estimating $\mu$. Robust mean estimation algorithms will solve this problem to within error $O(\eps)$, i.e., one can efficiently find $\tilde{\mu}$ s.t. $\|\mu-\tilde{\mu}\|_2 = O(\eps)$ and this is the best possible bound. However, suppose the center of the cube is the origin, and the estimated center has all coordinates equal to $\eps/\sqrt{d}$. Then, the TV distance between the two corresponding cubes is $1-(1-\frac{\eps}{\sqrt{d}})^d \simeq 1 - e^{-\eps\sqrt{d}}$, very far from the best possible TV distance of $O(\eps)$. On the other hand, if one could estimate the mean with $O(\eps)$ error in $L_1$ norm, this would result in a TV distance bound of $O(\eps)$. This is simply because one can bound the TV distance as the sum over the distances along the marginals, and for each marginal it is bounded by the distance between the means. However, estimating the mean within $L_1$ error $\eps$ is impossible in general, e.g., for a Gaussian. Indeed, almost all the recently developed methods in robust statistics naturally provide guarantees for mean estimation in $\ell_2$ norm and yield no useful guarantees in our setting. 
 \begin{figure}[h]
 \centering
    \includegraphics[width = 0.9\textwidth]{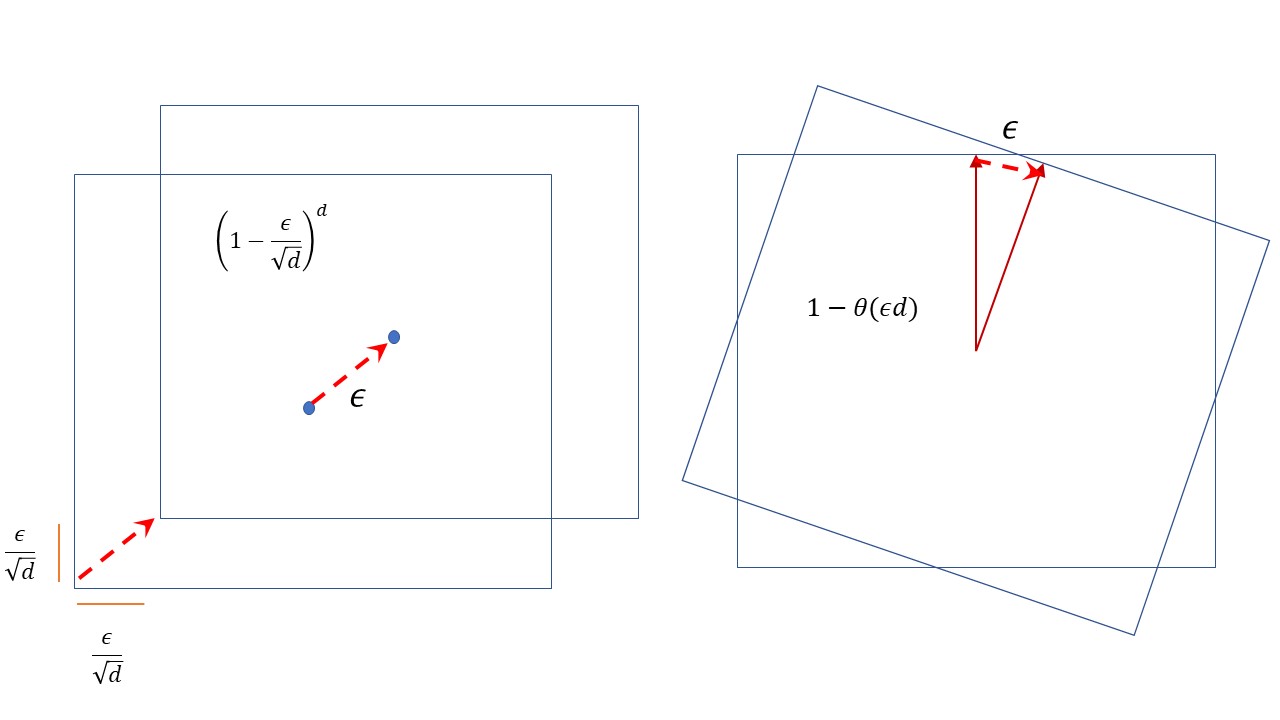}
	\caption{Shifting a cube in Euclidean norm or rotating can create very large TV distance.}
	\label{fig:L1vsL2}
	\end{figure}

The robust covariance and higher-moment estimation methods of ~\cite{KothariSS18} can be used to robustly learn the columns of the unknown linear transformation each to within $\eps$ error in Euclidean norm (after normalizing the covariance to be the identity). However, $\eps$ error in each column means a TV distance error of up to $1-(1-O(\eps))^d$, again growing with the dimension.  

\paragraph{Our Result:} The main contribution of this paper is developing an algorithm for learning affine transformations that circumvents the inherent issues with the method of moments and manages to obtain, using polynomial time and samples, almost optimal recovery guarantees in total variation distance. Specifically, our main result is a polynomial-time algorithm to robustly estimate an unknown affine transformation of the hypercube to within TV distance $O(\eps)$.

\begin{theorem}\label{thm:main}
Given $Y = \{y_1,\ldots, y_n\}$, an $\eps$-corrupted sample of points from an unknown parallelopiped $H = A[-1,1]^d+a$ in $\R^d$, there is a polynomial-time algorithm that outputs a parallelopiped $\widehat{H}=\widehat{A}[-1,1]^d+\widehat{a}$ s.t. $d_{TV}(\hat{H},H) = O(\eps)$.   
\end{theorem}

Our algorithm is based on a new approach to ICA, even without outliers, that does not rely on constant-order moments as in prior works. As we discuss next, a key contribution of our work is the development of a new \emph{certificate} of correctness of the estimated linear transformation that does not rely on low order moments. Our algorithm relies on an iterative update step that makes progress whenever the current guess on the unknown affine transformation fails a check in the certificate.

\subsection{Approach and techniques}

Let us first assume the affine transformation consists of only a shift and a diagonal scaling. In this case, we start with a coarse approximation of the center and side lengths obtained by using the coordinate-wise median, and a scaling of the interval in each coordinate that contains the middle $(1/2)+\eps$ fraction of samples. We then refine this iteratively using the fact that density of the cube is uniform along each coordinate, and, crucially, that the measure of the intersection of two axis-parallel bands is bounded by the product of their individual measures. If the latter condition is violated, then the intersection has a large fraction of corrupted samples, and we simply delete all the points in the intersection and continue. A simple and important idea here is that most of the corrupted points can be partitioned among the coordinate directions.  

Now consider a general rotation, i.e., $A$ is orthonormal and $a=0$. This turns out to be substantially more challenging. The following bound on the TV distance serves as a starting point.
\begin{lemma}\label{fact:dtv-rotation}
Let $H=[-1,1]^d$. Suppose $A, \widehat{A}$ are $d \times d$ matrices; $A$ is orthonormal and $\widehat{A}$ is a matrix with unit length rows. There is an absolute constant $C$ s.t. 
\[
d_{TV}(\widehat{A}H, AH)\le C \sum_i\norm{\widehat{a}_i-a_i}_2.
\]
\end{lemma}
This lemma follows from Lemma~\ref{fact:bound-out}, in which we prove the fraction of the volume of $\{x\in AH: x\cdot \widehat a_i >1\}$ is tightly bounded by $C\norm{\widehat{a_i}-a_i}_2$.
Crucially, the RHS terms are not squared. If they were, it would be the Frobenius norm and we can hope to robustly estimate to within low error. This norm however is not rotationally invariant, depends on the target basis (as it should!), and could be much larger. 

To learn a rotation, we start with a coarse approximation, where each facet normal is approximated to within poly($\eps$) error. Then, we consider one vector at a time (keeping the others fixed) and iteratively ``improve" it. Our desired objective is to maximize the number of uncorrupted sample points inside the parallelopiped defined by the current vectors. But this is hard to estimate or improve locally. Instead, we focus on minimizing the number of points outside the band $|\widehat{a}_i^\top x|\le 1$. We do this by proving that if the current distance $\|\widehat{a}_i-a_i\|_2 = \delta$, then the mean of points in the difference in this direction i.e., 
$S = \{x: \widehat{a}_i^\top x \ge 1, |a_i^Tx| \le 1\}$ gives us an indication of which direction to move $\widehat{a}_i$ to make it closer to $a_i$. In other words, we can improve the objective of number of points outside the band by a local update. However, the presence of outliers complicates matters, as a small number of outliers could radically alter the location of the mean outside the band. To address this challenge, we prove that estimating the mean of this subset $S$ {\em robustly in Euclidean norm} suffices to preserve the gradient approximately! Alongside, to keep the influence of noise under control, we ensure that pairwise intersections of bands are all small. Roughly speaking, our algorithm is robust gradient descent with provable guarantees. 

We combine the above procedures by alternating between them to robustly learn arbitrary affine transformations.

\section{Preliminaries}

\subsection{Robust Estimation}
\begin{theorem}[Robust Mean Estimation for Bounded Covariance Distributions]\cite{DKKLMS16}
\label{thm:robust-mean}
There exists a polynomial time
algorithm that takes input an $\eps$-corruption sample of a collection $X$ of $n$ points in $\R^d$ where the mean of $X$ is $\mu$ and the covariance of $X$ is $\Sigma$ and outputs an estimate $\tilde\mu$ satisfying
\[
\norm{\tilde \mu -\mu}_2 \le O(\sqrt{\eps})\norm{\Sigma}_2^{1/2}.
\]
\end{theorem}

\begin{theorem}[Theorem 1.4 in~\cite{KS17}, see also~\cite{LaiRV16}]
\label{thm:robust-ica}
There exists a polynomial time algorithm that given a corrupted sample $X$ of n points in $\R^d$ drawn from a rotated cube $AH$ where $A\in\R^{d\times d}$ is a orthogonal matrix with rows $a_1,\dots,a_d$, outputs component estimates $\hat a_1,\dots,\hat a_d\in\R^d$ with the following guarantee:
the components estimates satisfy with high probability, there exists a permutation $\pi\in S_d$ such that for any $i\in[d]$,
\[
\langle \hat a_i \cdot a_{\pi(i)}\rangle^2\ge 1-O(\sqrt{\eps}).
\]
\end{theorem}

\subsection{Logconcave functions}

The following lemmas are either from  \cite{LV07} or direct consequences.

\begin{lemma}\label{lem:intersection_of_bands}
Let $f:\R^2 \rightarrow \R_+$ be an isotropic two-dimensional logconcave density with associated measure $\nu$. Let $u,v$ be unit vectors in $\R^2$ with $|u^\top v| \le \frac{1}{2}$. Consider the bands $H_u = \{x:  u^\top x \ge a_u\}$ and $H_v = \{x: v^\top x \ge a_v\}$ and assume that the marginal densities along $u$ and $v$, $f_u$ and $f_v$ satisfy
$f_u(a_u), f_v(a_v)\ge 1/2$.  
Then,  for a universal constant $C$, 
\[
\nu(H_u \cap H_v) \le C \nu(H_u)\nu(H_v) 
\]
\end{lemma}

\begin{lemma}\label{lem:positive_mean}
Let $f$ be a one-dimensional logconcave density with mean $0$ and variance $\sigma^2$. Let $\mu^+ = \E_f(y|y\ge 0)$. Then there exist universal constants $c_1,c_2$ s.t. 
\[
c_1 \sigma \le \mu^+ \le c_2 \sigma.
\]
\end{lemma}

\begin{lemma}\label{lem:logconcave-balance-weight}
For any one-dimensional logconcave density $f:\R\rightarrow \R_+$ with mean $\mu$, we have 
\[
\Pr(X \ge \mu) \ge \frac{1}{e}.
\]
\end{lemma}

\begin{lemma}\label{lem:logconcave-mean-density}
Let $f$ be a one-dimensional logconcave density with mean $\mu$. Then 
\[
f(\mu)\ge \frac{1}{8}\max f(x).
\]
\end{lemma}

\begin{lemma}\label{lem:logconcave-moments}
Suppose $M_n(f)$ are the moments of $f$, i.e.,
\[
M_n(f) = \int_0^{\infty} t^n f(t) \; dt.
\]
If $f$ is logconcave, then the sequence $M_n(f)/n!$ is logconcave.
\end{lemma}
\subsection{Cubes}

The following facts about cubes will be useful. While the precise constants in the bounds are not important for our analysis, the first two facts have been a subject of inquiry in asymptotic convex geometry. 
\begin{lemma}\cite{barthe2003extremal}
\label{lem:slab-vol-upper-bound}
Let $B=[-1/2,1/2]^d$ be the $d$-dimensional cube with volume 1 and $a\in \R^d$ be an arbitrary unit vector. Then for all $t\le 3/4$,
\[
\vol(\{x\in B: |x \cdot a|> \frac{t}{2}\})\le 1-t.
\]
\end{lemma}

\begin{lemma}\label{lem:section-vol}\cite{ball1986cube}
Let $B=[-1/2,1/2]^d$ be the $d$-dimensional cube with volume 1 and $a\in \R^d$ be an arbitrary unit vector.. Then $(d-1)$-volume of any sections of $B$ defined by $a$ and $t$ is at most $\sqrt{2}$
\[
\vol_{n-1}(\{x\in B :x\cdot a =t\})\le \sqrt{2}.
\]
\end{lemma}

The next fact follows from the above lemma and a calculation using the logconcavity of one-dimensional marginals of the hypercube.

\begin{lemma}\label{lem:truncated-cube-mean}
Let $x\sim[-1,1]^d$ be a uniform random vector on $H=[-1,1]^d$ and $y=x\cdot a$ where $a$ is an arbitrary unit vector. Let $\mu_t = \E_f(y|y\ge t)$. For $0\le t\le 1/2$, there exists universal constant $c$ s.t. 
\[
\mu^t \ge t + c.
\]
\end{lemma}

\section{Robustly learning a Shift and Diagonal Scaling}\label{sec:shift-and-scaling}

In this section, we give an algorithm for robustly learning arbitrary shifts and diagonal scalings of the uniform distribution on the solid hypercube $H=[-1,1]^d$. That is, $x \rightarrow Ax + b$ when $A$ is a diagonal matrix. Since the uniform distribution on $H$ is symmetric around $0$, we can WLOG assume that all the entries of $A$ are non-negative. This special case is equivalent to learning affine transformations that correspond to a shift (i.e., introducing a non-zero mean) and scaling (i.e., scales each of the coordinates of the hypercube by an unknown and potentially different positive scaling factor). 

More precisely, we will prove: 

\begin{theorem}\label{thm:shift-and-scaling}
Suppose $H$ is an unknown axis-aligned cube in $\R^d$, that is, $H=AH_0+b = \otimes_{i=1}^d [u_i,v_i]$ where $A\in\R^{d\times d}$ is a diagonal matrix, $b\in\R^d$ and $H_0=[-1,1]^d$ is the unit cube. There exists an algorithm that, for small enough constant $\epsilon>0$, takes an $\epsilon$-corruption $X = x^{(1)}, x^{(2)}, \ldots, x^{(n)}$ of size $n \ge n_0 = \poly(d, 1/\eps)$ of an iid sample from the uniform distribution on $H$ and outputs $\widehat H=\widehat AH_0+\widehat b$ such that
\[
d_{TV}(H,\widehat H) \le 4\eps.
\]
\end{theorem}

\begin{remark}
While we will omit this refinement here, a more careful analysis of our algorithm produces a tighter bound of $d_{TV}(H,\widehat H) \le (2+\eps)\eps$.
\end{remark}

We first describe our algorithm:
\begin{algorithm}
\label{alg:shift-and-scaling}
\begin{enumerate}
\item \textbf{Input: } An $\epsilon$-corruption $X$ of an iid sample of size $n$ chosen from $\otimes_{i=1}^d [u_i,v_i]$. 
\item \textbf{Robust Range Finding: } For each $1 \leq i \leq d$, arrange the input corrupted sample in increasing order of $y_i$. Let $[l_i,r_i]$ be the interval of the smallest length that includes the middle $1/2+\epsilon$ fraction of the points (notice that such an interval is uniquely defined). Set $\hat{u}_i = 2l_i$ and $\hat{v}_i = 2r_i$.
Notice that $\{x \in \R^d: v_i \le x_i \le u_i \}$ is an axis-aligned cube that contains the true cube and all the side lengths are at most twice the true cube. 
\item \label{step:1d-check}
\textbf{One-Dimensional Density Check: } For all $i\in[n],k\in[d]$, check if 
\begin{align}
|S_{i,k}^+| :=\left|\left\{x:x^{(i)}\ge\left(1-\frac{k\eps}{d}\right)\hat{u}_i+\frac{k\eps}{d}\hat{v}_i\right\}\right|\ge\frac{k\eps}{2d}n \label{eqn:one-coord-max}\\
|S_{i,k}^-| :=\left|\left\{x:x^{(i)}\le \frac{k\eps}{d}\hat{u}_i+\left(1-\frac{k\eps}{d}\right)\hat{v}_i\right\}\right|\ge\frac{k\eps}{2d}n \label{eqn:one-coord-min}
\end{align}
\item \textbf{Update: } If (\ref{eqn:one-coord-max}) is false for some $i$ and $k$, update $\hat{u}_i=\hat{u}_i-\frac{k\eps}{d}(\hat{u}_i-\hat{v}_i)$ and go back to Step \ref{step:1d-check}.
If (\ref{eqn:one-coord-min}) is false for some $i$ and $k$, set $\hat{v}_i=\hat{v}_i+\frac{k\eps}{d}(\hat{u}_i-\hat{v}_i)$ and iterate (go back to Step \ref{step:1d-check}).
\item \textbf{Two-Dimensional Density Check: } For all $i,j\in[n],k_1,k_2\in[d]$, check if 
\begin{align}
|S_{i,k_1}^+\cap S_{j,k_2}^+|,|S_{i,k_1}^+\cap S_{j,k_2}^-|,|S_{i,k_1}^-\cap S_{j,k_2}^-|\le\frac{10k_1k_2\eps^2}{d^2}n\label{eqn:two-coord}
 \end{align}
\item \label{step:update-remove}
\textbf{Update: } 
If (\ref{eqn:two-coord}) is false for some $i,j$ and $k_1,k_2$, remove all points in the violating intersection from the sample set and go back to Step \ref{step:1d-check}.
\item \textbf{Return: } Output the cube $\otimes_{i=1}^d [\hat{u}_i, \hat{v}_i]$.
\end{enumerate}
\end{algorithm}

Notice that in this case, our goal is effectively to determine the intervals $[u_i, v_i]$ for each $1 \leq i \leq d$ that describe the $i$-th dimension of the shifted and scaled hypercube $H$. The key idea of the algorithm is a \emph{certificate} that checks a set of efficiently verifiable conditions on the corrupted sample with the guarantee that 1) the true parameters satisfy the checks, and 2) any set of parameters that satisfy the checks yield a hypercube $\hat{H}$ that is $O(\epsilon)$ different in symmetric volume difference (that gives us a total variation bound) from the true hypercube $H$. 

Our certificate itself is simple and natural and corresponds to checking that 1) for each of the coordinates $1 \leq i \leq d$, the discretized density (i.e., fraction of points lying in discrete intervals of size $\sim \epsilon/d$ in the purported range estimated from the corrupted sample matches the expected density of the uniform distribution on $AH+b$, and 2) for each pair of coordinates, the fraction of points in the intersection of intervals of length $\sim \epsilon/d$ along all possible pairs of directions match the expected density. Notice that all such statistics in an uncorrupted sample match those of the population with high probability so long as the the sample is of size $\poly(d/\epsilon)$. 

To analyze the above algorithm, we will prove that any set of parameters that satisfy the checks in our certificate (\ref{eqn:one-coord-max}),(\ref{eqn:one-coord-min}),(\ref{eqn:two-coord}) in Algorithm~\ref{alg:shift-and-scaling} must yield a good estimate of true parameters. Specifically, Lemma~\ref{lemma:small-outside} shows that the volume of $\widehat H$ that is not contained in $H$ is small.  Lemma~\ref{lemma:small-inside} shows that the volume of $H$ that is not covered by $\widehat H$ is small. Together, these two lemmas imply that our estimated distribution is $O(\epsilon$) close in total variation to the true hypercube. 
Our proof relies on an elementary combinatorial claim about set systems that we state next. We defer the proof to Section~\ref{sec:sum-of-intersections}.
\begin{restatable}{lemma}{intersectionsum}
\label{lem:intersection-sum}
Let $S_1,\dots,S_d \subseteq [n]$ be arbitrary subsets. Let $\mathrm{frac}(S):=|S|/n$ be normalized size of $S$. Suppose that 
\begin{equation}\label{eqn:union}
\mathrm{frac}\left(\bigcup S_i\right) = \eps
\end{equation}
for some $\eps<1$ and for all $i,j$,
\begin{equation}\label{eqn:intersection}
\mathrm{frac}(S_i\cap S_j)\le \alpha \mathrm{frac}(S_i)\mathrm{frac}(S_j),
\end{equation}
for some $\alpha>0$ s.t. $\alpha\eps <1$. Then,
\begin{equation}\label{eqn:sum}
\sum_{i\in[d]}\mathrm{frac}(S_i) \le \frac{\eps}{1-\alpha\eps}.
\end{equation}
\end{restatable}

\begin{lemma}\label{lemma:small-outside}
Let $X = \{x^{(1)}, x^{(2)}, \ldots, x^{(n)}\}$ be an $\epsilon$-corruption of an iid sample from the uniform distribution on $H = \otimes_{i = 1}^d [u_i,v_i]$. Let $\widehat H = \otimes_{i = 1}^d [\hat{u}_i, \hat{v}_i]$ be any axis-aligned hypercube satisfying (\ref{eqn:one-coord-max}),(\ref{eqn:one-coord-min}),(\ref{eqn:two-coord}). Then, with probability at least $1-1/d$ over the draw of the original uncorrupted sample,
\[
\frac{\vol(\widehat H \setminus H)}{\vol(\widehat H)}\le 4\eps.
\]
Here, $\vol$ denotes the usual Lebesgue volume of sets in $\R^d$.
\end{lemma}

\begin{proof}
For only the sake of analysis, we set $\hat u_i-\hat v_i=1$. In this normalization, $\vol(\widehat H)=1$. 
For each $i$, discretize the interval $[\hat{u}_i, \hat{v}_i]$ to a grid with intervals of size $\eps(\hat u_i-\hat v_i)/d=\eps/d$. We further assume that the vertices of the unknown $H$ are rounded to points in this grid. This assumption amounts to a change in the volume of $H$ by at most $\epsilon \vol(\hat{H})$. 
We will then show that
the density of points in the difference $\widehat H\setminus H$, i.e., $\{x: x \in 
\widehat H, x \notin H\}$ is 
\begin{equation}\label{eqn:density}
\frac{|\{x \in X:x\in\widehat H \setminus H\}|}{\vol(\widehat H \setminus H)}\ge \frac{1}{4}n.
\end{equation}

From the inequality above, it follows that $\vol(\widehat H \setminus H)\le \frac{4}{n}|\{x:x\in\widehat H\setminus H\}|\le 4\eps$
since $|\{x:x\in\widehat H\setminus H\}|\le\eps n$. So it suffices to prove (\ref{eqn:density}).

For all $i \in [d]$, suppose $u_i-\hat u_i=c_i\eps/d,\hat v_i-v_i=c_i'\eps/d$. Let $S_i =\{x \in X: \hat u_i\le x_i\le u_i,v_i\le x_i\le \hat v_i\}$. 
By (\ref{eqn:one-coord-max}) and (\ref{eqn:one-coord-min}), we have $|S_i|\ge (c_i+c_i')\frac{\eps}{2d}n$.
By (\ref{eqn:two-coord}), we have for any $i \neq j$,
\[
|S_i\cap S_j|\le 10 (c_i+c_i')\frac{\eps^2}{d^2} n \le 40 |S_i||S_j|n.
\]
Then by Lemma~\ref{lem:intersection-sum} with $\alpha=40$, 

\begin{align*}
    \frac{|\bigcup S_i|}{n-40 |\bigcup S_i|}&\ge \sum_{i\in [d]}\frac{|S_i|}{n},\\
    \frac{|\{x:x\in\widehat H \setminus H\}|}{n-40|\{x:x\in\widehat H \setminus H\}|}
&\ge \frac{\eps}{2d}\sum_{i=1}^d(c_i+c_i').
\end{align*}

By the upper bound on the total number $|\{x:x\in\widehat H\setminus H\}|\le\eps n\le n/80$, the denominator is at least $n/2$. Then
\[
|\{x:x\in\widehat H \setminus H\}|\ge\frac{\eps}{4d}n\sum_{i=1}^d(c_i+c_i').
\]
Thus we get (\ref{eqn:density}) by $\vol(\widehat H \setminus H)\le \frac{\eps}{d}\sum_{i=1}^d (c_i+c_i')$. 
\end{proof}

Next we show that Step~\ref{step:update-remove} removes at most $\eps n$ true points in total.

\begin{lemma}\label{lemma:delete-noise}
With high probability, at least half of points removed in Step~\ref{step:update-remove} are outliers.
\end{lemma}
\begin{proof}
Suppose we remove points from an intersection $S_{i,k_1}^+\cap S_{j,k_2}^+$, i.e., 
\begin{equation}\label{eqn:bad-intersection}
|S_{i,k_1}^+\cap S_{j,k_2}^+| > \frac{10k_1k_2\eps^2}{d^2}n.
\end{equation}
If the violating intersection is outside the true cube, all the points are outliers. Otherwise, we can compute the volume by definition of $S_{i,k}^+$
\[
\vol(S_{i,k_1}^+\cap S_{j,k_2}^+) = \frac{k_1k_2\eps^2}{d^2}(\hat u_i -\hat v_i)(\hat u_j-\hat v_j) \le \frac{4 k_1k_2\eps^2}{d^2}
\] 
where the last equality follows from $\hat u_i-\hat v_i$ is at most twice of the true side length. Then
with high probability, the number of original uncorrupted sample in the region $S_{i,k_1}^+\cap S_{j,k_2}^+$ of volume $4k_1k_2\eps^2/d^2$ is at most $5k_1k_2\eps^2n/d^2$. Thus, comparing with (\ref{eqn:bad-intersection}), we can see that at least half of the point removed by the algorithm are outliers.
\end{proof}

\begin{lemma}\label{lemma:small-inside}
Let $X = \{x^{(1)}, x^{(2)}, \ldots, x^{(n)}\}$ be an $\epsilon$-corruption of an iid sample from the uniform distribution on $H = \otimes_{i = 1}^d [u_i,v_i]$. Let $\widehat H = \otimes_{i = 1}^d [\hat{u}_i, \hat{v}_i]$ be an axis-aligned hypercube satisfying (\ref{eqn:one-coord-max}),(\ref{eqn:one-coord-min}),(\ref{eqn:two-coord}). Then, with probability at least $1-1/d$ over the draw of the original uncorrupted sample, 
\[
\frac{\vol( H \setminus \widehat H)}{\vol(H)}\le 4\eps.
\]
\end{lemma}
\begin{proof}
For only the sake of analysis, we set $\hat u_i-\hat v_i=1$ and $\vol(\widehat H)=1$. 
For $i\in[d]$ s.t. $\hat u_i < u_i$ or $v_i<\hat v_i$, we define $S_i = \{x:\hat u_i< x_i\le u_i, v_i\le x_i <\hat v_i\}$ and $c_i=u_i - \hat u_i,c_i' = \hat v_i - v_i$. Then 
\[
H\setminus \widehat H = \bigcup_{\{i:\hat u_i < u_i, \text{ or } v_i<\hat v_i\}} S_i.
\]
Since $S_i$ are removed by the algorithm, by (\ref{eqn:one-coord-max}) and (\ref{eqn:one-coord-min}), $|S_i|< \frac{1}{2}(c_i+c_i')n$. Since $S_i$ is in the true cube, the original uncorrupted sample in $S_i$ is at least $(c_i+c_i')n$ with high probability. So there are at least half of the original uncorrupted sample in $S_i$ are removed by either the adversary or by the algorithm. Suppose $\eta_i n$ is the number of original uncorrupted sample removed in $S_i$. Then $\eta_i\ge \frac{1}{2}(c_i+c_i')$.
By Lemma~\ref{lemma:delete-noise}, the total number of points deleted by the algorithm is at most $\eps n$. Hence the number of deleted points in $\bigcup S_i$ is at most $2\eps n$. The deletion in the intersection of $S_i$ and $S_j$ is naturally upper bounded by the number of original uncorrupted sample. So we can apply Lemma~\ref{lem:intersection-sum} with $\alpha=4$ to the deletion in $S_i$ and get 
\[
\frac{2\eps}{1-8\eps} \ge \sum_{\{i:\hat u_i < u_i, \text{ or } v_i<\hat v_i\}} \eta_i.
\]
If $\eps\le 1/16$, the left hand side is at most $4\eps$. Thus
\begin{align*}
\vol(H\setminus \widehat H) = \vol\left(\bigcup S_i \right) &\le \sum \vol(S_i)\\
&\le \frac{1}{2}\sum(c_i+c_i')\\
&\le \sum \eta_i\\
&\le 4\eps.
\end{align*}
\end{proof}

Now, with Lemmas~\ref{lemma:small-outside} and \ref{lemma:small-inside}, we can prove our main result of the shift and scaling case.
\begin{proof}[Proof of Theorem~\ref{thm:shift-and-scaling}]
If the algorithm outputs a cube $\widehat H$ satisfies (\ref{eqn:one-coord-max}),(\ref{eqn:one-coord-min}),(\ref{eqn:two-coord}), then by Lemma~\ref{lemma:small-outside} and Lemma~\ref{lemma:small-inside},
\[
d_{TV}(H,\widehat H)=1-\frac{\vol(H\cap\widehat H)}{\max\{\vol(H),\vol(\widehat H)\}}=\max\left\{\frac{\vol( H \setminus \widehat H)}{\vol(H)},\frac{\vol(\widehat H \setminus H)}{\vol(\widehat H)}\right\}=4\eps.
\]
Otherwise, the algorithm improves one of $\hat u_i$ or $\hat v_i$ by at least $\eps/d$. Thus, within at most $O(d^2/\eps)$ iterations, the algorithm terminates.
If the algorithm terminates when it deletes $2\eps n$ points, by Lemma~\ref{lemma:delete-noise}, we remove all noisy points.
\end{proof}

\section{Rotation}
\label{sec:rotation}
In this section, we give a robust algorithm to learn rotations of the uniform distribution of the standard hypercube. Specifically, we will show:

\begin{theorem}\label{thm:rotation}
Suppose $H = [-1,1]^d$ is the standard cube in $\R^d$ and $A\in \R^{d\times d}$ be an unknown rotation matrix. There exists an algorithm that given an $\epsilon$-corruption $X=\{x^{(1)},\dots,x^{(n)}\}$ of size $n \ge n_0 = \poly(d, 1/\eps)$ of an iid sample from the uniform distribution on $AH$, runs in poly($n$) time and outputs $\widehat A$ such that
\[
d_{TV}(AH,\widehat{A} H) = O(\eps).
\]
\end{theorem}

We first describe the algorithm:
\begin{algorithm}
\label{alg:rotation}

\textbf{Input}: $\epsilon$-corrupted sample $X =\{x^{(1)}, x^{(2)}, \ldots, x^{(n)}\}$. 

\textbf{Output}: $\widehat A$ with unit length rows $\hat a_{(1)},\dots,\hat a_{(d)}$.
\begin{enumerate}
\item \textbf{Warm Start: } Run the robust moment estimation algorithm in Theorem~\ref{thm:robust-ica} and get an estimate of $A$ with rows $a_{(1)}^*,\dots,a_{(d)}^*$, so that for each $i$, $\norm{\hat{a}_{(i)}-a_{(i)}^*}_2= O(\sqrt{\eps})$. Initialize $a_{(i)}^0 = \hat{a}_{(i)}$ for each $i$. 
\item  For any vector $a \in \R^d$, let $S(a)=\{x:x\cdot a>1\}\cup \{-x:x\cdot a<-1\}$. For $t = 1,2,\ldots, 2^{12} d\log d$,  do:
\begin{enumerate}
\item \textbf{Two-Dimensional Density Check: } For all $j<i$, check if 
\[
\frac{1}{n}|S(a_{(i)}^t)\cap S(\hat a_{(j)})| > \frac{2c_1}{n^2}|S(a_{(i)}^t)||S(\hat a_{(j)})|
\]

If false, remove all points in the intersection $S(a_{(i)}^t)\cap S(\hat a_{(j)})$.
\item \textbf{Robust GD: } Robustly estimate the mean $\tilde\mu$ by applying algorithm in Theorem~\ref{thm:robust-mean} to the subset $S(a_{(i)}^t)$. Set
\begin{align*}
\beta &= \frac{c_2|S(a_{(i)}^t)|}{\norm{\tilde\mu}^2n},\\
a_{(i)}^{t+1} &= \frac{a_{(i)}^t-\beta\tilde\mu}{\norm{a_{(i)}^t-\beta \tilde\mu}}\\
s^{t+1}& = \left|\{x:|x\cdot a_{(i)}^{t+1}|>1\}\right|
\end{align*}
Add back points that are removed in Step (a), i.e., recover the sample set to the origin one $X$.
\end{enumerate}
Let $\hat a_{(i)} = a_{(i)}^t$ where $t$ minimizes $s^{t}$.
\item Output the matrix $\widehat A$ with rows $\hat a_{(1)},\dots,\hat a_{(d)}$.
\end{enumerate}
\end{algorithm}

\subsection{Outline of analysis}
For the purpose of analysis, we assume that the true rotation matrix is the identity matrix in the rest of Section~\ref{sec:rotation}.
We will analyze Algorithm~\ref{alg:rotation} by the following propositions. First we show that we can separately learn the rows of the rotation matrix. Let $S(a)=\{x:x\cdot a>1\}\cup \{-x:x\cdot a<-1\}$. For a set of unit vectors $\hat a_{(1)},\dots,\hat a_{(d)}$ and a corrupted sample $X$, after we remove samples in all the intersections such that
\[
\frac{1}{n}|S(\hat a_{(i)})\cap S(\hat a_{(j)})| > \frac{2c_1}{n^2}|S(\hat a_{(i)})||S(\hat a_{(j)})|
\]
as in Step 2(a) of Algorithm~\ref{alg:rotation},
we can define $\eps_i$ as the fraction of outliers in $S(\hat a_{(i)})$ and $\eta_i$ as the fraction of original uncorrupted sample in $S(\hat a_{(i)})$ that are removed by the adversary or by the algorithm. That is, $\eps_i$ and $\eta_i$ are exactly the fraction of corrupted samples when we robustly estimate $\tilde\mu$. In the following Proposition, we show that the sum of $\eps_i$ and $\eta_i$ is small, which implies we do not increase the fraction of outliers by separately learning $a_{(1)},\dots,a_{(d)}$.
\begin{prop}
\label{prop:sum-of-noise}
Let $X_i$ be the subset of corrupted sample $X$ after removing samples in Step 2(a) of Algorithm~\ref{alg:rotation} and 
$S(\hat a_{(i)})=\{x:x\cdot \hat a_{(i)}>1, x \in X_i\}\cup \{-x:x\cdot \hat a_{(i)}<-1, x\in X_i\}$.
Let $\eps_i$ be the fraction of outliers in $S(\hat a_{(i)})$,
and $\eta_i$ be the fraction of original uncorrupted sample in $S(\hat a_{(i)})$ that are removed by the adversary or the algorithm.
Suppose for all $i$, $\hat a_{(i)}\cdot e_i\ge 0.99$ and for all $i,j$,
\[
\frac{1}{n}|S(\hat a_{(i)})\cap S(\hat a_{(j)})| \le \frac{2c_1}{n^2}|S(\hat a_{(i)})||S(\hat a_{(j)})|.
\]
Then,
\[
\sum_{i\in[d]}(\eps_i+\eta_i)\le 2\eps.
\]
\end{prop}

The proof of Proposition~\ref{prop:sum-of-noise} is deferred to Section~\ref{sec:separate-coordinate}.

Next, we give a robust algorithm that updates a single rotation vector $\hat a_{(i)}$ of the cube with a small constant fraction of corruption in $\{x:|x\cdot \hat a_{(i)}|>1\}$. We show that by repeatedly applying this algorithm on $\hat a_{(i)}$ at most $O(d\log d)$ times, we can learn the rotation vector $e_i$ with error $O(\eps_i+\eta_i)$. 

\begin{prop}
\label{prop:update-rotation-vec}
Suppose $H= [-1,1]^d$ is the true hypercube. 
There exists a robust algorithm that given $\eps$-corrupted samples from $H$ and a unit vector $a_{(i)}^0$ such that $\norm{a_{(i)}^0-e_i}_2\le 0.1$, outputs 
\[
\norm{\hat a_{(i)}-e_i}_2\le 2^{14}(\eps_i+\eta_i)
\]
after at most $2^{12} d\log d$ iterations, where $\eps_i+\eta_i$ is the fraction of corruption in $\{x:x\cdot \hat a_{(i)}>1\}$.
\end{prop}
The proof of Proposition~\ref{prop:update-rotation-vec} is deferred to Section~\ref{sec:update-rotation-vec}.
With the two propositions above, we complete the proof of Theorem~\ref{thm:rotation} in Section~\ref{sec:proof-of-rotation}.

\subsection{Properties of Cubes}
In this section, we prove a few basic properties of uniform samples from  $H = [-1,1]^d$, in particular, of the subset of samples $\{x\in H: x \cdot a > 1\}$ for a unit vector $a$. We will later use them in our analysis of Algorithm~\ref{alg:rotation}.

The following lemma is crucial for our analysis since it connects the 2-norm error of $\hat a_{(i)}$ and the volume of $\widehat A H\setminus AH$. Hence it implies Lemma~\ref{fact:dtv-rotation}.

\begin{lemma}\label{fact:bound-out}
Let $Y$ be a uniform sample from $H$ of size $n$. Suppose $a$ is a unit vector such that $\norm{a-e_i}=\delta$. Then with high probability the number of samples in $S(a) \cap Y$ is 
\[
\frac{\delta}{5} n 
\le |S(a) \cap Y| :=\left|\left\{x \in Y:|x\cdot a|> 1\right\}\right|
\le \frac{(1+\delta)\delta}{4} n.
\]
\end{lemma}

\begin{proof}
Let $\bar a$ be the $d-1$-vector that excludes $i$-th coordinate of $a$. Since $\norm{a-e_i}=\delta$, we have that $a_{i} = a \cdot e_i= 1-\delta^2/2$ and $\norm{\bar a} = \delta + O(\delta^2)$. Let 
$S^+(a) := \{x:x\cdot a> 1\}$.
By symmetry of the cube, $2\vol(S^+(a))=\vol(S(a))$. So it suffices to compute the number of original uncorrupted sample in $S^+(a)$ and double it.
We can compute the volume of $S^+(a)$ by integrating the length of $x_i$ over $\bar x$ where $\bar x$ is the vector $x$ without entry $x_i$. The upper bound on $\bar x$ is 1 and the lower bound is $x\cdot a=1$, i.e., $x_i=\frac{1-\bar x\cdot \bar a}{a_i}$. So we have
\[
\vol(S^+(a)) = \int_{\bar x \in [-1,1]^{d-1}} \max \left\{0, 1- \frac{1-\bar x\cdot \bar a}{a_i} \right\}\; d \bar x
\]
Let $y=\bar x \cdot \bar a$. Suppose $p(y)$ is the density of $y$. Let $1- \frac{1-\bar x\cdot \bar a}{a_i}\ge 0$. We have $y\ge \delta^2/2$. Thus, we can write the integral above as 
\begin{equation}
\vol(S^+(a)) = 2^{d-1}\left(
\Pr\left(y\ge\frac{\delta^2}{2}\right)\left(1-\frac{1}{a_i}\right)+\frac{1}{a_i}\int_{y\ge \frac{\delta^2}{2}} y \cdot p(y)\; dy
\right).
\end{equation} 
Since $\bar x$ is a uniform random variable on $[-1,1]^{d-1}$, we have $y$ is a scaled 1d-projection of the standard $(d-1)$-cube onto an arbitrary direction with zero mean and variance $\norm{\bar a}^2/3=\delta^2/3$. Since $y$ is logconcave and symmetric, we can bound the probability $1/2-\delta\le \Pr\left(y\ge\frac{\delta^2}{2}\right)\le 1/2$. We define $\mu_y:=\E(y \mid y\ge \delta^2/2)=\int_{y\ge \frac{\delta^2}{2}} y \cdot p(y)\; dy$ as the truncated mean of $y$. Then
\begin{equation}
\label{eqn:vol-s1-rotation}
\vol(S^+(a)) = 2^{d-1} \Pr\left(y\ge\frac{\delta^2}{2}\right) \left( 1-\frac{1}{a_i} +\frac{\mu_y}{a_i}
\right)
\end{equation}
By Lemma~\ref{lem:positive_mean}, $\mu_y$ is upper and lower bounded by $\Theta(\delta)$. The upper bound is achieved when $\bar a$ is in the direction of true coordinates, which is 
 $(1+\delta)\delta/2$. Let $y^+$ be the truncated density of $y$ on $\R^+$. Then 
 raw moments of $y^+$ are $M_0(y^+) = 1,M_1(y^+)\le \mu_y,M_2(y^+) = \delta^2/3$.
By Lemma~\ref{lem:logconcave-moments}, $M_1^2\ge M_0\frac{M_2}{2}$. Then $\mu_y\ge M_1\ge \frac{1}{\sqrt 6}\delta \ge \frac{2}{5}\delta$.
Plugging the bounds of $\Pr\left(y\ge\frac{\delta^2}{2}\right)$ and $\mu_y$ into (\ref{eqn:vol-s1-rotation}), we get
\begin{align*}
\frac{2\delta}{5}-\frac{\delta^2}{2} &\le 1-\frac{1}{a_i} +\frac{\mu_y}{a_i}
=\frac{\mu_y-\delta^2/2}{1-\delta^2/2} \le \frac{\delta+\delta^2}{2}\\
\frac{2^{d-1}}{5}\delta n &\le \vol(S^+(a))\le 2^{d-3}(1+\delta)\delta n.
\end{align*}
Since the volume of the standard cube $H$ is $2^d$, the probability a uniformly random sample is in $S^+(a)$ is at least $2^{d-1}\delta/(5\cdot 2^d)=\delta/10$ and at most $2^{d-3}(1+\delta)\delta/2^d=(1+\delta)\delta/8$. Thus the fraction of samples in $S(a)$ is at least $\delta/5$ and at most $(1+\delta)\delta/4$.
\end{proof}

\begin{lemma}\label{fact:mean-proj-bound}
Suppose $\norm{a-e_i}=\delta$.
Let $\mu$ be the mean of the subset of the cube: $\{x\in[-1,1]^d:a\cdot x>1\}$. Then
$\mu\cdot e_i \le 1-\delta/32$.
\end{lemma}
\begin{proof}
Let $S^+(a)=\{x\in[-1,1]^d:x \cdot a>1\}$ and $S' = \{x\in[-1,1]^d:x\cdot a >1, x_i = x\cdot e_i >1-c\delta\}$ for a fixed constant $c<1$. Then by the trivial upper bound: $x_i\le 1$ for any $x$ in the cube,
\[
\mu\cdot e_i\le \frac{\vol(S') + (1-c\delta)(\vol(S^+(a))-\vol(S'))}{\vol(S^+(a))}
= 1-c\delta\left(1-\frac{\vol(S')}{\vol(S^+(a))}\right).
\]
We can compute the volume of $S'$ by integral over $x_i$.
\begin{align*}
\frac{\vol(S')}{\vol(H)} &= \frac{1}{2}\int_{1-c\delta}^{1}\Pr(x\cdot a >1\mid x_i)\;dx_i\\
&=\frac{1}{2}\int_{1-c\delta}^{1}\Pr(\bar x\cdot \bar a\ge 1-x_i\cdot a_i)\;dx_i
\end{align*}
Let $y=\bar x\cdot \bar a/\norm{\bar a}$. Since $x_1$ and $\bar x$ are independent, $\bar x$ is uniformly random over the $d-1$ standard cube and $y$ is the projection of $\bar x$ onto an arbitrary unit vector. Let $t = (1-x_i)/\delta$ and hence $x_i = 1-t\delta$.
Then $dx_i=-\delta\; dt$. By substitution,
\begin{align*}
\frac{\vol(S')}{\vol(H)}
= \frac{1}{2}
\int_{t=0}^{c}\Pr\left (y\ge \frac{1-x_i\cdot a_i}{\norm{\bar a}}\right)\delta\;dt
\le \frac{1}{2}
\int_{t=0}^{c}\Pr(y\ge t)\delta\;dt
\end{align*}
By Lemma~\ref{lem:slab-vol-upper-bound}, $\Pr(y\ge t)\le 1/2(1-t)$ for any $t\le 3/4$.  
Take $c=1/2$. Then we get
\[
\frac{\vol(S')}{\vol(H)} 
\le \frac{1}{2}\int_{t=0}^{c}\frac{1}{2}(1-t)\delta\;dt
=\frac{1}{4}(c-c^2/2)\delta = \frac{3}{32}\delta.
\]
By Lemma~\ref{fact:bound-out}, $\frac{\vol(S^+(a))}{\vol(H)}\ge\delta/10$. Thus
\[
\mu\cdot e_i
= 1-c\delta\left(1-\frac{\vol(S')}{\vol(S^+(a))}\right)\le1-\delta/32.
\]
\end{proof}

\begin{lemma}\label{fact:variance-in-coordinate}
Let $H$ be the standard hypercube $[-1,1]^d$. Suppose $\norm{a-e_i}=\delta$. Then the variance of uniform distribution on $H \cap \{x:x \cdot a >1\}$ in the direction $e_i$ is at most $c\delta^2$.
\end{lemma}

\begin{proof}
Suppose $p(x_i)$ is the density function of $x_i$ on $\{x:x\cdot a>1\}$ and $\bar x,\bar a$ are vectors constructed from $x,a$ by excluding $i$-th coordinate.  Then 
\[
p(x_i)\propto \vol_{d-1}(\{x:x\cdot a>1,x_i \text{ is fixed}\})
\propto\Pr_{\bar x}(\bar x\cdot \bar a>1-a_i x_i).
\]
Let $y=\bar x\cdot \bar a$. Then $y$ is logconcave and symmetric with mean 0 and variance $\delta^2/3$. So in the range $x_i\in[-1,1]$, $p(x_i)$ is increasing. Then 
\[
\max p(x_i)=p(1)\propto \Pr(y>\delta^2/2)\ge \frac{1}{2}-\delta.
\]
and 
\[
p(1-\delta)\propto \Pr(y>1-a_i(1-\delta))\le \Pr(y>\delta)\le \frac{3-2\sqrt{2}}{4}.
\]
If $\delta< 0.1$, then $p(1-\delta)<\frac{1}{8}\max p(x_i)$. 
By Lemma~\ref{lem:logconcave-mean-density}, $p(\mu)\ge\frac{1}{8}\max p(x_i)$. By the monotonicity of $p(x_i)$, we have $\mu\ge 1-\delta$. Then by Lemma~\ref{lem:positive_mean}, the variance of $x_i$ on the whole subset is upper-bounded by the variance of one side of the mean. 
\end{proof}

\subsection{Rotation matrix can be learned coordinate by coordinate}
\label{sec:separate-coordinate}
In this section, we prove Proposition~\ref{prop:sum-of-noise} via Lemmas \ref{lemma:sum-of-noise-is-small} and \ref{lemma:sum-of-removal-is-small}. Recall that $S(\hat a_{(i)})=\{x:x\cdot \hat a_{(i)}>1\}\cup \{-x:x\cdot \hat a_{(i)}<-1\}$ and $\eps_i$ is the fraction of outliers in $S(\hat a_{(i)})$, $\eta_i$ is the fraction of original uncorrupted sample in $S(\hat a_{(i)})$ that are removed. 

\begin{lemma}\label{lemma:sum-of-noise-is-small}
Let $X_i$ be the subset of corrupted sample $X$ after removing samples in Step 2(a) of Algorithm~\ref{alg:rotation} and 
$S(\hat a_{(i)})=\{x:x\cdot \hat a_{(i)}>1, x \in X_i\}\cup \{-x:x\cdot \hat a_{(i)}<-1, x\in X_i\}$.
Let $\eps_i$ be the fraction of outliers in $S(\hat a_{(i)})$.
Suppose for all $i$, $\hat a_{(i)}\cdot e_i\ge 0.99$ and
for all $i,j$,
\[
\frac{1}{n}|S(\hat a_{(i)})\cap S(\hat a_{(j)})| \le \frac{2c_1}{n^2}|S(\hat a_{(i)})||S(\hat a_{(j)})|.
\]
Then
$\sum_{i\in[d]}\eps_i\le 2\eps$.
\end{lemma}

\begin{proof}
Let $S_N(\hat a_{(i)})$ be the set of outliers that in the halfspace $\{x:x\cdot \hat a_{(i)}>1\}$. We know the fraction of all outliers is at most $\eps$. So we have $\mathrm{frac}\left(\bigcup S_N(\hat a_{(i)})\right) \le \eps$. We may assume that $\eps_i\ge c_3\mathrm{frac}(S(\hat a_{(i)}))$. Otherwise, we can reduce $\mathrm{frac}(S(\hat a_{(i)}))$ by running the update algorithm in Section~\ref{sec:update-rotation-vec} (Corollary~\ref{coro:update-rotation-vec}). Then 
\[
\mathrm{frac}(S_N(\hat a_{(i)})\cap S_N(\hat a_{(j)}))\le \mathrm{frac}(S(\hat a_{(i)})\cap S(\hat a_{(j)}))\le 2c_1\mathrm{frac}(S(\hat a_{(i)}))\mathrm{frac}(S(\hat a_{(j)}))\le\frac{2c_1}{c_3^2} \eps_i\eps_j.
\]
Set $\alpha=\frac{2c_1}{c_3^2}$, we get 
\[
\mathrm{frac}(S_N(\hat a_{(i)})\cap S_N(\hat a_{(j)})) \le \alpha \mathrm{frac}(S_N(\hat a_{(i)}))\mathrm{frac}(S_N(\hat a_{(j)})).
\]
Then by Lemma~\ref{lem:intersection-sum}, $\sum \eps_i=\sum_{i\in[d]}\mathrm{frac}(S_N(\hat a_{(i)}))\le\frac{\eps}{1-\alpha\eps}$. If $\eps\le \frac{c_3^2}{4c_1}$, then $1-\alpha\eps > 1/2$ and hence $\sum \eps_i\le 2\eps$.
\end{proof}

\begin{lemma}\label{lemma:sum-of-removal-is-small}
Let $X_i$ be the subset of corrupted sample $X$ after removing samples in Step 2(a) of Algorithm~\ref{alg:rotation} and 
$S(\hat a_{(i)})=\{x:x\cdot \hat a_{(i)}>1, x \in X_i\}\cup \{-x:x\cdot \hat a_{(i)}<-1, x\in X_i\}$.
 Let $\eta_i$ be the fraction of original uncorrupted sample in $S(\hat a_{(i)})$ that are removed by the adversary or the algorithm.
Suppose for all $i$, $\hat a_{(i)}\cdot e_i\ge 0.99$ and
for all $i,j$,
\[
\frac{1}{n}|S(\hat a_{(i)})\cap S(\hat a_{(j)})| \le \frac{2c_1}{n^2}|S(\hat a_{(i)})||S(\hat a_{(j)})|
\]
then
$\sum_{i\in[d]}\eta_i\le 2\eps$.
\end{lemma}
\begin{proof}
If $\norm{\hat a_{(i)}-e_i} = \delta_i\ge 8\eps_i'$ and $\norm{\hat a_{(j)}-e_j} = \delta_j\ge 8\eps_j'$ where $\eps_i'$ is the fraction of original uncorrupted sample deleted by the adversary, then by Lemma~\ref{lem:intersection_of_bands}, the number of original uncorrupted sample in $\{x:x\cdot \hat a_{(i)}>1,x\cdot \hat a_{(j)}>1\}$ is at most $c_1\delta_i\delta_j$. The number of points that are deleted by the algorithm in $\{x:x\cdot \hat a_{(i)}>1,x\cdot \hat a_{(j)}>1\}$ is at least $2c_1(\delta_i-\eta_i)(\delta_j-\eta_j)$. Then the fraction of original uncorrupted sample in the deleted points is at most 
\[
\frac{c_1\delta_i\delta_j}{2c_1(\delta_i-\eps_i')(\delta_j-\eps_j')}\le \frac{1}{2(1-1/8)^2}\le 32/49.
\]
In this case, at least 17/49 of the points that are deleted by the algorithm are outliers. So the number of deleted original uncorrupted sample is at most $\frac{32}{17}\eps<2\eps$. 

If $\norm{\hat a_{(i)}-e_i} < 8\eps_i'$ (or $\norm{\hat a_{(j)}-e_j} < 8\eps_j'$), then we can upper bound the total number of original uncorrupted sample in $\{x:x\cdot \hat a_{(i)}>1\}$ by $8\eps_i'/4$. So the number of true points we deleted in this case is at most $\frac{8}{4}\sum_i\eps_i'\le 2\eps$.
\end{proof}

Combining Lemmas~\ref{lemma:sum-of-noise-is-small} and \ref{lemma:sum-of-removal-is-small} gives the conclusion of Proposition~\ref{prop:sum-of-noise}.

\subsection{Update rotation vector by the mean of outside samples}
\label{sec:update-rotation-vec}

As in the previous sections, we assume $a$ is the current rotation vector and $e$ is the true rotation vector such that $\norm{a-e}=\delta$. We first give a non-robust algorithm that improves $a$. Given a sample $X=\{x^{(1)},\dots,x^{(n)}\}$ and a unit vector $a$, the algorithm computes $\mu$ of the subset of the sample
\[
\{x: x\cdot a >1\}\cup\{-x: x\cdot a <-1\}
\]  and outputs 
\[
a'= \frac{a-\beta\mu}{\norm{a-\beta\mu}}.
\]

\begin{lemma}\label{lemma:non-robust-update}
If the step size $\beta \le 2c\delta/\norm{\mu}^2$, then the non-robust algorithm outputs a unit vector $a'$ such that $a'\cdot e - a\cdot e\ge \beta\delta/32$.
\end{lemma}
\begin{proof}
By the definition of $a'$, we have
\[
a'\cdot e = \frac{a\cdot e -\beta \mu\cdot e}{\sqrt{(a-\beta \mu)(a-\beta\mu)}}
\]
Since $\norm{a-e}=\delta$ and $a,e$ are unit vectors, $a\cdot e = (2-\norm{a-e}^2)/2 = 1-\delta^2/2$. By Fact~\ref{fact:mean-proj-bound}, we know $\mu\cdot a\ge 1+c\delta$ and $\mu\cdot e\le 1-c\delta$. Plugging into the right hand side of the equation above, we get
\begin{align*}
a'\cdot e
& \ge \frac{1-\delta^2/2 -\beta(1-c\delta)}{\sqrt{1-2\beta(1+c\delta) +\beta^2\norm{\mu}^2}}\\
& \ge \frac{1-\beta+c\beta\delta-\delta^2/2}{\sqrt{1-2\beta}}\\
&\ge (1-\beta+\beta\delta/32-\delta^2/2)(1+\beta+\frac{3}{2}\beta^2)\\
&\ge 1-\delta^2/2+\beta\delta/32
\end{align*}
where the second inequality follows from our assumption that $\beta \le 2c\delta/\norm{\mu}^2$. Thus $a'\cdot e - a\cdot e\ge \beta\delta/32$.
\end{proof}

\begin{algorithm}
\label{alg:update-step-non-robust}
\textbf{Input}: corrupted sample $x_1,\dots,x_n$ and a unit vector $a$ such that $\norm{a-e}=\delta$.

\textbf{Output}: a unit vector $a'$

\textbf{Parameters}: step size $\beta>0$
\begin{enumerate}
\item \label{step:mean} Compute the robust mean $\tilde\mu$ of the following set using the algorithm in Theorem~\ref{thm:robust-mean} 
\[
\{x: x\cdot a >1\}\cup\{-x: x\cdot a <-1\}
\]
\item Output 
\[
a'= \frac{a-\beta\tilde \mu}{\norm{a-\beta \tilde\mu}}
\]
\end{enumerate}
\end{algorithm}

Then the robust version of the algorithm is to replace the mean $\mu$ in Step~\ref{step:mean} by the robust mean $\tilde\mu$. We will use the robust mean estimation algorithm for bounded covariance distributions in Theorem~\ref{thm:robust-mean}.

Suppose the number of corrupted sample in the subset $\{x:x\cdot a >1\}$ is $\eps_i n$. With the similar argument as in the proof of Lemma~\ref{lemma:non-robust-update}, we can show the following result.
\begin{lemma}
\label{lemma:robust-update-step}
    If $\eps_i\le 2^{-14}\delta$ and $\beta\le \frac{\delta}{32\norm{\tilde\mu}^2}$, then the robust version of the algorithm outputs a unit vector $a'$ such that $a'\cdot e_i - a \cdot e_i \ge \frac{1}{128}\beta\delta$.
\end{lemma}

\begin{proof}
By the definition of $a'$, we have
\[
a'\cdot e = \frac{a\cdot e -\beta \tilde\mu\cdot e}{\sqrt{(a-\beta \tilde\mu)(a-\beta\tilde\mu)}}
\]
Since $\norm{a-e}=\delta$ and $a,e$ are unit vectors, $a\cdot e = (2-\norm{a-e}^2)/2 = 1-\delta^2/2$. 
\begin{align*}
a'\cdot e_i
& \ge \frac{1-\delta^2/2 -\beta\tilde\mu\cdot e}{\sqrt{1-2\beta\tilde \mu \cdot a +\beta^2\norm{\mu}^2}}\\
&\ge (1-\delta^2/2 -\beta\tilde\mu\cdot e)(1+\beta\tilde\mu\cdot a -\beta^2\norm{\tilde\mu}^2/2+O(\beta^2))\\
& \ge 1-\delta^2/2+\beta(\tilde\mu\cdot a - \tilde\mu\cdot e)-\beta^2\norm{\tilde\mu}^2/2-\frac{\delta^2}{2}\beta\tilde\mu\cdot a 
\end{align*}
By Fact~\ref{fact:mean-proj-bound}, we know $\mu\cdot e_i\le 1-\delta/32$. Assuming that the fraction of corruption is $c$, we apply the algorithm in Fact~\ref{thm:robust-mean} to compute the robust mean $\tilde\mu$, which gives an error guarantee that $(\tilde\mu-\mu)\cdot e_i\le O(\sqrt{c})\sigma(e_i)$ where $\sigma^2(e_i)$ is the variance of $\{x:x\cdot a>1\}$ in the direction $e_i$. By Lemma~\ref{fact:variance-in-coordinate}, $\sigma^2(e_i)$ is at most $\delta$.
Thus we have $\tilde\mu\cdot e\le 1-(1/32-\sqrt c)\delta$. By Lemma~\ref{fact:bound-out}, the fraction of points in  $\{x:x\cdot a>1\}$ is at least $(\delta/5)n$. So the fraction of corruption is bounded by
\[
c \le \frac{\eps_i}{\delta/5}\le 5\cdot 2^{-14}
\]
where the last inequality follows from our assumption that $\eps_i\le 2^{-14}\delta$.
By the trivial upper bound $\tilde\mu\cdot a>1$ and our condition $\norm{\tilde\mu}^2\le \frac{\delta}{32\beta}$, we get 
\[
a'\cdot e_i - a\cdot e_i\ge 
\beta\left((\frac{1}{32}-\sqrt{c})\delta-\frac{\delta}{64}-\frac{\delta^2}{2}\right) \ge
\frac{1}{128}\beta\delta.
\]
\end{proof}

\begin{coro}
\label{coro:update-rotation-vec}
By repeatedly running Algorithm~\ref{alg:update-step-non-robust} $2^{12} d \log d$ times, one of $a_{(i)}^t$ satisfies 
\[
\norm{a_{(i)}^t-e_i}\le 2^{14}\eps_i.
\]
\end{coro}

\begin{proof}
We start with $\norm{a-e_i}=\delta$. Each time we run the algorithm, 
\[
a'\cdot e_i-a\cdot e_i\ge\frac{1}{128}\beta\delta=\frac{\delta^2}{128\cdot 32\norm{\tilde\mu}^2}\ge \frac{\delta^2}{2^{12}d}
\]
by Lemma~\ref{lemma:robust-update-step} and the upper bound $\sqrt{d}$ on norm of the mean. Then $\norm{a'-e}/\norm{a-e}\le\sqrt{1-\frac{1}{2^{11}d}}$. The optimal error we can get by repeatedly running this algorithm is $2^{10}\eps_i$ from the condition of Lemma~\ref{lemma:robust-update-step}. When $\norm{a-e_i}>2^{14}\eps_i$, we have $\eps_i\le 2^{-14}\delta$ so the algorithm will output a better $a'$ as in Lemma~\ref{lemma:robust-update-step}. Suppose after $t$ steps, the error is reduced to $2^{14}\eps_i$. Then
\[
\left(\sqrt{1-\frac{1}{2^{11}d}}\right)^t \delta = 2^{14}\eps_i.
\]
Since $\delta\le \sqrt{\eps}$ and $\eps_i\ge \eps/d$, we get 
\[
t\le 2^{12} \cdot d\cdot \log\frac{\delta}{2^{14}\eps_i}\le 2^{12} d\log d.
\]
\end{proof}

\subsection{Proof of Theorem~\ref{thm:rotation}}
\label{sec:proof-of-rotation}

Let $\eps_i^t$ be the fraction of outliers in the set $\{x:x\cdot a_{(i)}^t>1\}$ and $\eta_i^t$ be the fraction of original uncorrupted sample removed from the set $\{x:x\cdot a_{(i)}^t>1\}$.

\begin{proof}[Proof of Theorem~\ref{thm:rotation}]
By Corollary~\ref{coro:update-rotation-vec}, we know there exists an $a_{(i)}^t$ for all $i$ such that $\norm{a_{(i)}^t - e_i}\le 2^{10}(\eps_i^t+\eta_i^t)$.
First we show that $A=(a_{(1)}^t,\dots,a_{(d)}^t)^\top$ is a good rotation matrix that $AH$ is close to $H$ in TV distance.
By Lemmas~\ref{lemma:sum-of-noise-is-small} and \ref{lemma:sum-of-removal-is-small}, 
\[
\sum_i(\eps_i^t+\eta_i^t) \le 4\eps.
\]
Let $S(a)=\{x:x\cdot a>1\}$ be a subset of samples after removing points in Step 2(a) of Algorithm~\ref{alg:rotation}. Then by Lemma~\ref{fact:bound-out},
\[
\sum_{i\in[d]} |S(a_{(i)}^t)| \le \sum_{i\in[d]}\frac{1}{4}\norm{a_{(i)}^t - e_i} n
\le \frac{1}{4}\sum_{i\in[d]}2^{10}(\eps_i^t+\eta_i^t)\le 2^{10}\eps n.
\]
Next, we prove $\hat A=(\hat a_{(1)},\dots,\hat a_{(d)})^\top$ is a good estimation of the true rotation matrix with respect to TV distance by showing the number of samples inside $\hat A H$ is at least $(1-c\eps)n$. The number of samples in $\hat AH$ is at least 
\[
n-\sum_i|S(\hat a_{(i)})| \ge n-\sum_i|S(a_{(i)}^t)| 
\ge (1-2^{10}\eps)n.
\]
where the first inequality follows from $|S(\hat a_i)|$ is the smallest over all $t$.
Since each iteration is polytime and the total number of iterations is $O(d^2 \log d)$, the overall complexity is polynomial in $d$.
\end{proof}

\section{General Case}
In this section, we give an algorithm to robustly learn an affine transformation of shift, scaling and rotation and prove our main result:
\label{sec:general}
\begin{theorem}
\label{thm:main-restate}
Given $X = \{x^{(1)},\dots,x^{(n)}\}$ an $\eps$-corrupted sample of points from an unknown cube $H=\{x\in\R^d: v_i^*\le a_{(i)}^*\cdot x\le u_i^*,\forall i\in[d]\}$, where $n \ge n_0 = \poly(d, 1/\eps)$, there exists a polynomial-time algorithm that outputs $\widehat H=\{x\in\R^d: \hat v_i\le \hat a_{(i)}\cdot x\le \hat u_i,\forall i\in[d]\}$ s.t. $d_{TV}(\widehat H,H)=O(\eps)$.
\end{theorem}

\begin{algorithm}
\label{alg:general}
\begin{enumerate}
\item Run the robust mean and covariance estimation algorithms in Theorem~\ref{thm:robust-mean} and get estimates of the mean $\hat \mu$ and covariance $\widehat \Sigma$ of the cube with error 
\begin{align*}
    \norm{\hat\mu-\mu^*} &\le C\sqrt{\eps}\norm{\Sigma^*}_2^{1/2}\\
    \norm{\widehat\Sigma-\Sigma^*}_2 &\le C\sqrt{\eps}\norm{\Sigma^*}_2.
\end{align*}
\item Run the robust moment estimation algorithm in Theorem~\ref{thm:robust-ica} and get an estimate of the rotation matrix $\widehat A$ with unit rows $\hat a_{(1)},\dots,\hat a_{(d)}$.

\item Repeat the following steps until the number of samples in $\{x: \hat v_i\le\hat a\cdot x\le \hat u_i,\forall i\}$ is at least $(1-c\eps)n$.
\begin{enumerate}
\item Run Algorithm~\ref{alg:shift-and-scaling} with $\hat a_{(1)},\dots,\hat a_{(d)}$ being the coordinates and update upper bounds and lower bounds of the cube $\hat u_1,\dots,\hat u_d,\hat v_1,\dots,\hat v_d$
\item Run Algorithm~\ref{alg:rotation} with upper bounds and lower bounds $\hat u_1,\dots,\hat u_d,\hat v_1,\dots,\hat v_d$ and update $\hat a_{(1)},\dots,\hat a_{(d)}$.
\end{enumerate}
\item Output the cube $\widehat H = \{x: \hat v_i\le\hat a_{(i)}\cdot x\le \hat u_i,\forall i\}$.
\end{enumerate}
\end{algorithm}

We analyze our algorithm by the following two propositions. Proposition~\ref{prop:alg1-with-rotaion-error} shows that the shift and scaling algorithm (Step 3(a)) learns the upper and lower bounds of the facets with error $O(\eps)+0.7\delta$ where $\delta$ is the error in rotation matrix. Proposition\ref{prop:alg2-with-scaling-error} shows that the rotation algorithm (Step 3(b)) learns the rotation matrix with error $O(\eps)+\delta$ where $\delta$ is the error in shift and scaling part.

\begin{prop}\label{prop:alg1-with-rotaion-error}
Suppose $\norm{\hat a_{(i)}-a_{(i)}^*}=\delta_i\le C\sqrt{\eps}$. Algorithm~\ref{alg:shift-and-scaling} outputs $\hat u_1,\dots,\hat u_d,\hat v_1,\dots,\hat v_d$ such that 
\[
\frac{|\hat u_i - u_i^*|+|\hat v_i-v_i^*|}{u_i^*-v_i^*} \le 44\eps_i + 0.7\delta_i
\]
and $\sum_{i\in[d]}\eps_i\le 2\eps$.
\end{prop}

\begin{prop}\label{prop:alg2-with-scaling-error}
Suppose $\frac{|\hat u_i - u_i^*|+|\hat v_i-v_i^*|}{u_i^*-v_i^*}= \delta_i \le C\sqrt{\eps}$. Algorithm~\ref{alg:rotation} outputs $\hat A$ with unit rows $\hat a_{(1)},\dots,\hat a_{(d)}$ such that 
\[
\norm{\hat a_{(i)}-a_{(i)}^*}\le 2^{20}\eps_i+ \delta_i
\]
where $\sum_{i\in[d]}\eps_i\le 2\eps$.
\end{prop}
We can now prove the main theorem.
For the purpose of analysis, we assume that the true cube is the standard cube $H=[-1,1]^d$ in the remaining of Section~\ref{sec:general}.

\begin{proof}[Proof of Theorem~\ref{thm:main-restate}.]
Let $\delta_{s,i} = |\hat u_i-1|+|\hat v_i-1|$ be the scaling error in coordinate $i$ and
$\delta_{r,i} \le \norm{\hat a_{(i)} -e_i}$ be the rotation error in coordinate $i$.
By Theorems~\ref{thm:robust-mean} and \ref{thm:robust-ica}, we start with $\delta_{s,i} = |\hat u_i-1|+|\hat v_i-1|\le C\sqrt{\eps}$ and $\delta_{r,i}=\norm{\hat a_{(i)}-e_i}\le C\sqrt{\eps}$. 
By Proposition~\ref{prop:alg1-with-rotaion-error}, Step 3(a) of Algorithm~\ref{alg:general} improves $\delta_{s,i}$ to $\delta_{s,i}'\le 44\eps_{s,i}+0.7\delta_{r,i}$. 
With this updated $\delta_{s,i}$, by Proposition~\ref{prop:alg2-with-scaling-error}, Step 3(b) improves $\delta_{r,i}$ to $\delta_{r,i}' \le c\eps_{r,i}+\delta_{s,i}'\le 2^{20}\eps_{r,i}+44\eps_{s,i}+0.7\delta_{r,i}$. If $\delta_{r,i}\ge 220 \eps_{s,i} + 5\cdot 2^{20}\eps_{r,i}$, we have $\delta_{r,i}'\le 0.9 \delta_{r,i}$. Otherwise, we have 
\begin{equation}
\label{eqn:rotation-error}
\sum_{i\in[d]}\norm{\hat a_{(i)}-e_i} =\sum_{i\in[d]}\delta_{r,i} \le \sum_{i\in[d]} 220\eps_{s,i}+5\cdot 2^{20}\eps_{r,i}\le (440+10\cdot 2^{20})\eps
\end{equation}
where the last equality follows from the upper bounds on the sum of $\eps_{s,i}$ and $\eps_{r,i}$ in Propositions~\ref{prop:alg1-with-rotaion-error} and \ref{prop:alg2-with-scaling-error}. Then \begin{equation}
\label{eqn:scaling-error}
\sum_{i\in[d]}|\hat u_i-1|+|\hat v_i-1|=\sum_{i\in[d]}\delta_{s,i}\le \sum_{i\in[d]} 44\eps_{s,i}+0.7\delta_{r,i}\le (396+7\cdot 2^{20})\eps.
\end{equation}
Thus, adding up (\ref{eqn:rotation-error}) and (\ref{eqn:scaling-error}), we get $d_{TV}(\widehat H,H)= O(\eps)$. Next we show that the number of iterations is bounded.
We start with $\delta_{r,i}\le C\sqrt{\eps}$ and end with $\delta_{r,i}\ge \eps/d$. Suppose we run Step 3 for $t$ iterations. Then 
$0.9^t C\sqrt{\eps}\le \eps/d$, which implies $t\le C\log(d/\eps)$.
\end{proof}

\subsection{Proof of Proposition~\ref{prop:alg1-with-rotaion-error}}

In this section, we will prove Proposition~\ref{prop:alg1-with-rotaion-error}. The idea is similar with the proof of Theorem~\ref{thm:shift-and-scaling} in Section~\ref{sec:shift-and-scaling}. But we need to deal with the error from the rotation matrix besides outliers.  The following two lemmas are general versions of Lemmas~\ref{lemma:small-outside} and \ref{lemma:small-inside} with rotation error.
\begin{lemma}
\label{lemma:small-outside-general}
Suppose a cube $\widehat H=\{x:\hat v_i\le \hat a_{(i)}\cdot x\le \hat u_i\}$ satisfies (\ref{eqn:one-coord-max}),(\ref{eqn:one-coord-min}),(\ref{eqn:two-coord}) in Algorithm~\ref{alg:shift-and-scaling} and $\delta_i=\norm{\hat a_{(i)}-a_{(i)}}$, then with high probability, for all $\hat u_i>1$,
\[
\hat u_i-1\le 0.35\delta+22\eps_i
\]
and for all $\hat  v_i<-1$,
\[
-1- \hat v_i\le 0.35\delta+22\eps_i
\]
where $\eps_i>0$ and $\sum_{\{i:\hat u_i>1 \text{ or } \hat v_i<-1\}} \eps_i\le 2\eps$.
\end{lemma}

\begin{proof}
Suppose $\delta_i=\norm{\hat a_{(i)}-e_i}$ is the error of rotation vector in direction $e_i$ and $\eps_i$ is the fraction of outliers in the set $\{x:1< \hat a_{(i)}\cdot x\le \hat u_i\}$ for $\hat u_i>1$. By (\ref{eqn:one-coord-max}), we know that the fraction of samples in $\{x:1< \hat a_{(i)}\cdot x\le \hat u_i\}$ is at least $(1/2)(\hat u_i-1)$. These points not in $\hat H$ are either outliers or because of the rotation error. By Lemma~\ref{fact:bound-out}, the fraction of original uncorrupted sample in $\{x:1< \hat a_{(i)}\cdot x\le \hat u_i\}$ is at most $\delta/8$. So we have
\[
\frac{1}{2}(\hat u_i -1) \le \delta_i/8 + \eps_i.
\]
Now we consider the case $\eps_i \le 0.05 \delta_i$. Then we get for $i$ s.t. $\hat u_i>1$ and $\eps_i\le 0.05\delta_i$,
\begin{equation}
\label{eqn:large-delta}
\hat u_i-1\le \delta_i/4 + 0.1\delta_i = 0.35\delta_i.
\end{equation}
The above inequality also applies to $v_i$ by symmetry.
Next if $\eps_i> 0.05\delta_i$, we have $(\hat u_i -1) \le 22\eps_i$. Set $\alpha=40\cdot22^2$. Then by (\ref{eqn:two-coord}), we have the fraction of the intersection of $\{x:1< \hat a_{(i)}\cdot x\le \hat u_i\}$ and $\{x:1< \hat a_{(j)}\cdot x\le \hat u_j\}$ is at most $\alpha \eps_i\eps_j$. Since the fraction of corruption is at most $\eps$, we can apply Lemma~\ref{lem:intersection-sum} to outliers in $\{x:1< \hat a_{(i)}\cdot x\le \hat u_i\}\cup\{x: \hat v_i\le \hat a_{(i)}\cdot x< -1\}$.
and get an upper bound of $\sum \eps_i$ 
\[
\sum_{\{i:\hat u_i>1,\eps_i \le 0.05 \delta_i\}}\eps_i
\le \frac{\eps}{1-\alpha\eps}.
\]
If $\alpha\eps\le 1/2$, we get $\sum_{\{i:\hat u_i>1,\eps_i \le 0.05 \delta_i\}}\eps_i\le 2\eps$. Then for $i$ s.t. $\hat u_i>1$ and $\eps_i\le 0.05\delta_i$,
\begin{equation}\label{eqn:small-delta}
\hat u_i-1 \le 22\eps_i.
\end{equation}
Add (\ref{eqn:large-delta}) and (\ref{eqn:small-delta}) gives the conclusion.
\end{proof}

\begin{lemma}\label{lemma:small-inside-general}
With high probability, Step 3(a) of Algorithm~\ref{alg:general} outputs $\hat u_i$ and $\hat v_i$ for all $\hat u_i<1$,
\[
1-\hat u_i\le 0.35\delta+22\eta_i
\]
and for all $\hat  v_i>-1$,
\[
\hat v_i-(-1)\le 0.35\delta+22\eta_i
\]
where $\eta_i>0$ and $\sum_{\{i:\hat u_i>1 \text{ or } \hat v_i<-1\}} \eta_i\le 2\eps$.
\end{lemma}

\begin{proof}
If $\hat u_i<1$, by (\ref{eqn:one-coord-max}), the fraction of points in $\{x:\hat u_i\le \hat a_{(i)}\cdot x\le 1\}$ is less than $\frac{1}{2}(1-\hat u_i)$. So at least half of points in the region are either removed (by adversary or by the algorithm) or because of the error from the rotation vector $\hat a_{(i)}$. By Lemma~\ref{fact:bound-out}, the fraction of error from rotation vector is at most $\delta_i/8$. 
\[
\frac{1}{2}(1-\hat u_i) \le \delta_i/8 +\eta_i.
\]
By Lemma~\ref{lemma:delete-noise}, the fraction of points are removed is at most $2\eps$, i.e., the fraction of the union of the removed points is at most $2\eps$. By the same argument in the proof of Lemma~\ref{lemma:small-outside-general}, we can upper bound $\sum \eta_i$ by $4\eps$ for $\eta_i\ge 0.05\delta$.
\end{proof}

\begin{proof}[Proof of Proposition~\ref{prop:alg1-with-rotaion-error}]
Lemma~\ref{lemma:small-outside-general} proves the conclusion in the cases that $\hat u_i>1$ or $\hat v_i<-1$. 
Lemma~\ref{lemma:small-inside-general} proves the conclusion in the cases that $\hat u_i<1$ or $\hat v_i>-1$.
\end{proof}

\subsection{Proof of Proposition~\ref{prop:alg2-with-scaling-error}}
In this section, we prove Proposition~\ref{prop:alg2-with-scaling-error}. The proof follows the same idea as the proofs of the rotation case in Section~\ref{sec:rotation}.
\begin{lemma}\label{fact:variance-in-coordinate-general}
Suppose $\norm{a-e_i}=\delta$. Then the variance of the subset of the standard cube $\{x:x \cdot a >1-\delta/2\}$ in the direction $e_i$ is at most $c\delta^2$ for a constant $c>0$.
\end{lemma}

\begin{proof}
Since the fraction of points in $\{x:x \cdot a >1-\delta/2\}$ that satisfies $x_i\ge 1-(3\delta/2)$ is at least $1-c\delta$,
by Lemma~\ref{lem:logconcave-balance-weight}, 
we know that the mean of $x_i$ on the subset $\{x:x \cdot a >1-\delta/2\}$ is at least $1-(3\delta/2)$. Then by Lemma~\ref{lem:positive_mean}, the variance of $x_i$ on the whole subset is upper-bounded by the variance of one side of the mean. 
\end{proof}

\begin{lemma}\label{fact:variance-dir-a-general}
Suppose $\norm{a-e_i}=\delta$. Then the variance of the subset of the standard cube $\{x:x \cdot a >1-\delta/2\}$ in the direction $a$ is at most $c\delta^2$ for a constant $c>0$.
\end{lemma}

\begin{proof}
Using the similar argument as in the proof of Lemma~\ref{fact:variance-in-coordinate-general}, it suffices to show that the mean of $x\cdot a $ on the subset $\{x:x \cdot a >1-\delta/2\}$ is at most $1+\delta/2$, which follows from Lemma~\ref{lem:logconcave-balance-weight} and the fact that the fraction of points in $\{x:x \cdot a >1-\delta/2\}$ that satisfies $ x\cdot a \le 1+(3\delta/2)$ is at least $1-c\delta$.
\end{proof}

\begin{lemma}
\label{lem:tail-bound-dir-a}
Suppose $\norm{a-e_i}=\delta$. The fraction of original uncorrupted sample of the standard cube in $\{x:|x\cdot a|>1+\delta/2\}$ is at least $\delta/64$ with high probability.
\end{lemma}

\begin{proof}
We can compute the volume of $S_1 =\{x \in[-1,1]^d:x\cdot a>1+\delta/2\}$ by integrating the length of $x_i$ over $\bar x$ where $\bar x$ is the vector $x$ without entry $x_i$. 
\[
\vol(S_1) = \int_{\bar x \in [-1,1]^{d-1}} \max \left\{0, 1- \frac{1+\delta/2-\bar x\cdot \bar a}{a_i} \right\}\; d \bar x
\]
Let $y=\bar x \cdot \bar a$. Suppose $p(y)$ is the density of $y$. We can write the integral above as 
\begin{equation}
\label{eqn:vol-s1}
\vol(S_1) = 2^{d-1}\left(
\Pr\left(y\ge\frac{\delta^2}{2}+\delta/2\right)\left(1-\frac{1+\delta/2}{a_i}\right)+\frac{1}{a_i}\int_{y\ge \frac{\delta^2}{2}+\delta/2} y \cdot p(y)\; dy
\right).
\end{equation} 
Since $\bar x$ is a uniform random variable on $[-1,1]^{d-1}$, we have $y$ is 1d projection of the standard $d-1$ cube onto an arbitrary direction with zero mean and variance $\norm{\bar a}^2/3=\delta^2/3$. Then we can bound the probability $1/8-\delta\le \Pr\left(y\ge\frac{\delta^2}{2}+\delta/2\right)\le 1/4$ and $\Pr(y\ge \frac{3}{4}\delta)\ge 1/8-\delta$. By Lemma~\ref{lem:truncated-cube-mean}, we have $\mu_y:=\E(y \mid y\ge \delta^2/2+\delta/2) \ge(\delta/2) + (\delta^2/2)+ c\delta$.
Plugging into the inequality (\ref{eqn:vol-s1}), we get 
\begin{align*}
\vol(S_1) &\ge 2^{d-1}\left( \Pr\left(y\ge\frac{\delta^2}{2}+\delta/2\right) \left( 1-\frac{1+\delta/2}{a_i} +\frac{\mu_y}{a_i}
\right)\right)\\
&\ge 2^{d-1}(1/8-\delta)(\delta/8+O(\delta^2))\\
&= 2^{d-7}\delta.
\end{align*}
Thus the probability a uniformly random sample is in $S_1$ is at least $2^{d-7}\delta/2^d=\delta/128$. By symmetry of the cube, the probability of a point that $x\cdot a < -1-\delta/2$ is the same. So the fraction of samples in $\{x:|x\cdot a|> 1+\delta/2\}$ is at least $\delta/64$.
\end{proof}

\begin{lemma}\label{lemma:true-mean-difference}
Suppose $\norm{a-e_i}=\delta$. Let $\mu$ be the mean of the subset of the cube $\{x\in [-1,1]^d: a\cdot x>1+\Delta\}$, where $\Delta\in[-\delta/2,\delta/2]$. Then $\mu\cdot a-\mu\cdot e_i\ge \delta/64$.
\end{lemma}

\begin{proof}
Let $S_1 =\{x \in[-1,1]^d:x\cdot a>1+\delta/2\}$. By the definition, we have $x\cdot a >1+\delta/2$ and $x\cdot e\le 1$ for all $x\in S_1$. So $\mu(S_1)\cdot(a-e)\ge \delta/2$. Then we will show that for $t\in[-\delta/2,\delta/2]$, the mean of the section of the cube $\mu(1+t):=\E(x\mid x\cdot a =1+t)$ in the direction $e_i$ is at most $1+t$. This is trivial for $t\ge 0$. So we assume $-\delta/2\le t<0$.
\begin{align*}
    \mu(1+t)\cdot e_i &= \E (x_i\mid x\cdot a =1+t)\\
    & = \int_{-1}^1 x_i \Pr(\bar x \cdot \bar a=1+t-a_ix_i)\;dx_i
\end{align*}
where $\bar x\sim[-1,1]^{d-1}$. Then $\bar x\cdot \bar a$ is symmetric at zero, that is, at $x_i = 1+t+O(\delta^2)$. So the truncated mean $\E(x_i\mid x\cdot a = 1+t, x_i\ge 1+2t)$ is $1+t$. Since $\E(x_i\mid x\cdot a = 1+t, x_i< 1+2t)<1+2t$, we conclude that the mean of these two parts $\E (x_i\mid x\cdot a =1+t)$ is at most $1+t$.

By Lemma~\ref{lem:tail-bound-dir-a}, the fraction of points in $S_1=\{x \in[-1,1]^d:x\cdot a>1+\delta/2\}$ is at least $\delta/64$. By Lemma~\ref{fact:bound-out}, the fraction of $\{x:x\cdot a>1\}$ is at most $\delta/8$. By Lemma~\ref{lem:section-vol}, the maximum of the section of the cube is $\sqrt 2$. Then the fraction of $\{x:1+t\le x\cdot a \le1\}$ is upper bounded by $\sqrt 2 t/2$. Thus
\[
\mu \cdot a -\mu \cdot e \ge 
\frac{\mu(S_1)\cdot (a-e)\frac{\delta}{32}}{(\delta/8)- (t/\sqrt {2}))}
\ge \delta/64.
\]
\end{proof}

\begin{proof}[Proof of Proposition~\ref{prop:alg2-with-scaling-error}]
Suppose $\delta_i = \norm{\hat a_{(i)}-e_i}$. In the proof of Lemma~\ref{lemma:robust-update-step}, we show that if $\tilde\mu\cdot \hat a_{(i)} - \tilde\mu\cdot e\ge c\delta_i$, then Algorithm~\ref{alg:update-step-non-robust} outputs a unit vector $\hat a_{(i)}'$ such that $\hat a_{(i)}'\cdot e_i-\hat a_{(i)}\cdot e_i\ge c'\beta\delta_i$. 
In the general algorithm, the only difference is we replace the true threshold in the direction $e_i$ by $\hat u_i$, i.e., $\tilde\mu$ is the robust mean of the subset $\{x: \hat a_{(i)}\cdot x >\hat u_i \}$. Let $\mu_i$ be the true mean of $\{x: \hat a_{(i)}\cdot x >\hat u_i \}$. 
By Lemma~\ref{lemma:true-mean-difference}, we have $\mu\cdot \hat a_{(i)}-\mu\cdot e_i\ge \delta_i/64$. 
We apply the algorithm in Theorem~\ref{thm:robust-mean} to compute the robust mean, which give an error guarantee that $(\tilde\mu-\mu)\cdot \hat a_{(i)}\le O(\sqrt c)\sigma(\hat a_{(i)})$ and 
$(\tilde\mu-\mu)\cdot e_i\le O(\sqrt c)\sigma(e_i)$ where $\sigma^2(\hat a_{(i)})$ is the variance of $\{x: \hat a_{(i)}\cdot x >\hat u_i \}$ in the direction $\hat a_{(i)}$ and $\sigma^2(e_i)$ is the variance in the direction $e_i$. By Lemmas~\ref{fact:variance-in-coordinate-general} and \ref{fact:variance-dir-a-general}, the variances are at most $\delta^2$. By Lemma~\ref{lem:tail-bound-dir-a}, the number of true points in $\{x: \hat a_{(i)}\cdot x >\hat u_i \}$ is at least $(\delta/64)n$. Thus if 
\[
c\le\frac{\eps_i}{\delta/64}\le 2^{14},
\]
the algorithm can improve $\hat a_{(i)}$ to $\hat a_{(i)}'$ such that
\[
\hat a_{(i)}'\cdot e_i-\hat a_{(i)}\cdot e_i\ge \frac{1}{256}\beta\delta_i.
\]
This implies that the algorithm make progress until $\delta < 2^20\eps_i$. Then the argument in the proofs of Corollary~\ref{coro:update-rotation-vec} and Theorem~\ref{thm:rotation} follows.
\end{proof}

\section{Total size of almost pairwise disjoint sets}
In this section, we prove Lemma~\ref{lem:intersection-sum}.
\label{sec:sum-of-intersections}
\intersectionsum*

To prove Lemma~\ref{lem:intersection-sum}, we will use the following technical result on almost pairwise independent Bernoulli random variables.
\begin{lemma}\label{lem:rv-sum}
Let $X=(X_1,X_2,\dots,X_d)\in\{0,1\}^d$ be a random variable on $\{0,1\}^d$. If for any $i,j\in[d]$, 
\[
\E X_iX_j\le \eps\E X_i \E X_j
\]
where $0\le\eps<1$, then 
\[
\sum_{i=1}^d\E X_i\le\frac{1}{1-\eps}.
\]
\end{lemma}
\begin{proof} By Jensen's inequality, we have:
\begin{align*}
\left(\sum_{i=1}^d\E X_i\right)^2
&=\left(\E\left(\sum_{i=1}^d X_i\right)\right)^2 \le \E\left(\sum_{i=1}^d X_i\right)^2\\
&= \sum_{i=1}^d\E X_i^2+\sum_{i,j\in[d]}\E X_iX_j\\
&\le \sum_{i=1}^d \E X_i + \eps\sum_{i,j\in[d]}\E X_i \E X_j \le \sum_{i=1}^d \E X_i + \eps\left(\sum_{i=1}^d\E X_i\right)^2
\end{align*}
Rearranging yields that $\sum_{i=1}^d\E X_i\le\frac{1}{1-\eps}.$
\end{proof}

We now proceed to the proof of Lemma~\ref{lem:intersection-sum}.
\begin{proof}[Proof of Lemma~\ref{lem:intersection-sum}]
Let $y$ be distributed uniformly on $\bigcup S_i$. Define a random variable $X=(X_1,X_2,\dots,X_d)\in\{0,1\}^d$ as a function of $y$ as follows:
\[
X_i = 
\begin{cases}
0 \quad\text{if } y \notin S_i\\
1 \quad\text{if } y \in S_i
\end{cases}
\]
Then by (\ref{eqn:union}), $\E X_i = \Pr(y\in S_i)=\frac{|S_i|}{|\bigcup S_i|}=\mathrm{frac}(S_i)/\eps$ and $
\E X_iX_j = \Pr(y\in S_i\cap S_j)=\frac{|S_i\cap S_j|}{|\bigcup S_i|}=\mathrm{frac}(S_i\cap S_j)/\eps$.

By (\ref{eqn:intersection}), we have 
\[
\E X_i X_j \le \frac{\size(S_i \cap S_j)}{\epsilon} \le \frac{\alpha}{\epsilon} \size(S_i) \size(S_j) =\alpha \eps \E X_i\E X_j.
\]
By Lemma~\ref{lem:rv-sum}, this implies $\sum_{i=1}^d\E X_i\le\frac{1}{1-\alpha \eps}.$
Using that
\[
\sum \E X_i = \frac{1}{\eps}\sum_{i=1}^d \mathrm{frac}(S_i)\,
\]
gives $\sum_{i=1}^d \mathrm{frac}(S_i)\le \frac{\eps}{1-\alpha \eps}$.
\end{proof}

\bibliographystyle{plain}
\bibliography{allrefs,stat_algs,ICA_bibliography,custom2,custom3,custom}

\section*{Appendix}

\begin{proof}[Proof of Lemma~\ref{lem:intersection_of_bands}.]
We can project the density to the span of $u,v$ to get a center-symmetric, isotropic, two-dimensional logconcave density $f:\R^2\rightarrow \R_+$. 
Let us assume that $u,v$ are orthogonal, the general case when they are at least at a constant angle will be similar.

Consider the level set $L_1$ of $f$ of function value at least  $1/100$. We claim that it intersects the line $\ell_1$ defined by $u^Tx =1$ in a segment of length at least $1/10$. Let $\delta=\nu(H_u)$. Moreover, the line $\ell_2$ defined by $u^Tx = 1+c\delta$ does not intersect $L_1$, else the measure in between the two lines is too large, and we know it is at most $\delta$. 

Now consider $L_2$, the level set of function value at least $1/1000$. Now we claim that the intersection of $\ell_2$ with $L_2$ is of length most $1/10$. Moreover, the line $\ell_3$ defined by $u^Tx = 1+2c\delta$ does not intersect $L_2$. 

The same bounds apply along $v$. 

To bound the measure of $\{x:\, u^Tx \ge 1, v^Tx \ge 1\}$, we divide up the region $u^T x \ge 1$ into strips perpendicular to $u$ so that the measure of the bounding lines decreases by a constant factor in each strip. Each strip has length $c\delta$ We do the same for $v$. So the intersection of the first two strips, a square $S$, has measure $O(\delta^2)$. Now, using the previous claims about level sets, it follows that along any line starting at the intersection of $u^Tx=1$ and $v^Tx =1$ and continuing in the region $u^Tx \ge 1, v^Tx \ge 1$, the value of $f$ decreases by a constant factor every $c_2 \delta$ distance along the line. From this it follows that the measure of the entire regions is $O(\delta^2)$. To see this we consider the polar integral of the region 
$H_u \cap H_v$,
\begin{equation*}
\int_{\theta = 0}^{\pi/2}\int_{r=0}^\infty r f(1+ r\cos\theta, 1+ r \sin\theta)\,dr\,d\theta \le \int_{\theta = 0}^{\pi/2}\int_{r=0}^\infty c_1re^{-r/(c_2\delta)}\,dr\,d\theta = O(\delta^2).
\end{equation*}

When $u$ and $v$ are at some constant angle (instead of orthogonal), then the intersection of bands induced by intervals beomes a parallelogram with area $O(\delta^2)$. The rest of the argument remains the same.
\begin{figure}
 \centering
    \includegraphics[width = 0.6\textwidth]{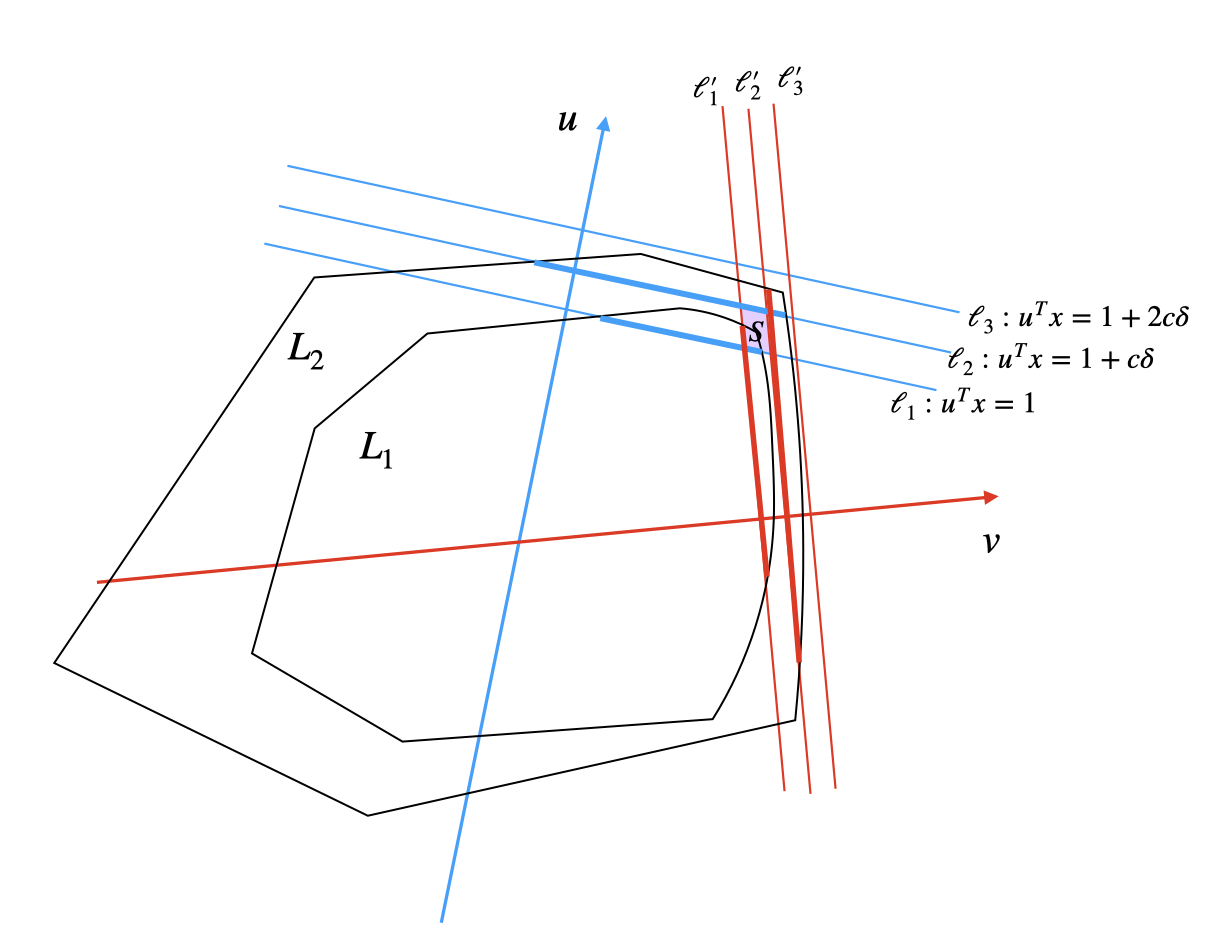}
	\caption{The intersection of halfspaces has measure $O(\delta^2)$.}
	\label{fig:intersection-measure}
	\end{figure}
\end{proof}

\end{document}